\documentclass[sigconf, nonacm]{acmart}

\usepackage{enumitem}

\usepackage{tikz}

\AtBeginDocument{}

\newcommand\blfootnote[1]{%
  \begingroup
  \renewcommand\thefootnote{}\footnote{#1}%
  \addtocounter{footnote}{-1}%
  \endgroup
}

\title{The Impact of Partial Computations on the Red-Blue Pebble Game}

\author{P\'al Andr\'as Papp}
\orcid{0009-0005-6667-802X}
\email{pal.andras.papp@huawei.com}
\affiliation{
  \department{Computing Systems Lab}
  \institution{Huawei Zurich Research Center}
  \city{Zurich}
  \country{Switzerland}
}

\author{Aleksandros Sobczyk}
\email{aleksandros.sobczyk@h-partners.com}
\orcid{0000-0002-1602-8329}
\affiliation{
  \department{Computing Systems Lab}
  \institution{Huawei Zurich Research Center}
  \city{Zurich}
  \country{Switzerland}
}

\author{Albert-Jan N. Yzelman}
\email{albertjan.yzelman@huawei.com}
\orcid{0000-0001-8842-3689}
\affiliation{
 \department{Computing Systems Lab}
  \institution{Huawei Zurich Research Center}
  \city{Zurich}
  \country{Switzerland}
}

\begin{document}

\begin{abstract}
We study an extension of the well-known red-blue pebble game (RBP) with partial computation steps, inspired by the recent work of Sobczyk~\cite{partial}. While the original RBP assumes that we need to have all the inputs of an operation in fast memory at the same time, in many concrete computations, the inputs can be aggregated one by one into the final output value. These partial computation steps can enable pebbling strategies with much smaller I/O cost, and in settings where such a step-by-step aggregation is possible, this extended red-blue pebble game offers a much more realistic cost model.

We establish the fundamental properties of this partial-computing red-blue pebble game (PRBP), and compare it to the original RBP. We begin with some simple examples where allowing partial computations can decrease the optimal I/O cost. It is also shown that the cost can decrease by up to a linear factor this way, but in general, it is NP-hard to decide whether partial computations allow for a smaller cost in a specific DAG. We then discuss how $S$-partitions, a crucial tool for deriving I/O lower bounds in RBP, can be adapted to the PRBP model. These new tools are then used to establish lower bounds on the I/O cost of some prominent computational tasks. Finally, we also adapt a hardness result from RBP, showing that the optimum cost is still NP-hard to approximate in PRBP to any reasonable factor.
\end{abstract}

\ccsdesc[300]{Theory of computation~Models of computation}
\ccsdesc[300]{Computer systems organization~Processors and memory architectures}

\keywords{Red-blue pebble game, partial computation, I/O cost, lower bounds}

\maketitle

\blfootnote{\copyright P\'al Andr\'as Papp, Aleksandros Sobczyk and Albert-Jan N. Yzelman, 2025. This is the author's full version of the work, posted here for personal use. Not for redistribution. The definitive version was published in the 37th ACM Symposium on Parallelism in Algorithms and Architectures (SPAA 2025), https://doi.org/10.1145/3694906.3743320.}

\section{Introduction}

The optimal execution of a complex computational task with limited resources is a fundamental question in computer science, with applications ranging from scientific simulations to machine learning. In many of these applications, the main performance bottleneck is in fact the vertical data movements between different layers of memory, known as I/O operations. One of the prominent tools to analyze this problem is the red-blue pebble game of Hong and Kung~\cite{RBpebbling1}, which captures the \emph{I/O complexity} of a computation in a two-level memory hierarchy: it studies the amount of vertical communication required when we want to execute the computation using a fast memory of limited size (e.g.\ cache), and a slow memory of much larger capacity (e.g.\ RAM). This model has been a vital tool to measure I/O complexity in different computations, and has been studied extensively in the past decades.

As in most pebble games, the computation in red-blue pebbling is modeled as a directed acyclic graph (DAG) $G=(V,E)$, with the nodes $V$ representing subtasks or operations, and a directed edge $(u,v) \in E$ representing that the output of operation $u$ is required as an input for operation $v$. We denote the number of nodes in $G$ by $n$, and the maximum input and output degree by $\Delta_{in}$ and $\Delta_{out}$, respectively. We assume that $G$ has no isolated nodes.

The red-blue pebble game, which we abbreviate by \emph{RBP}, assumes that we have a fast memory of some limited capacity $r$, and a slow memory of unlimited capacity. In order to compute a value, all of its inputs need to be in fast memory. During the game, a red pebble placed on a node $v$ represents the fact that that the output data of $v$ is currently kept in fast memory, and a blue pebble on $v$ represents that it is currently kept in slow memory. In the initial state, only the source nodes of $G$ have blue pebbles on them. A pebbling strategy can then apply the following transition rules:
\begin{enumerate}[label=\arabic*), leftmargin=1.5em, topsep=4pt, itemsep=2pt]
 \item \textsc{save}: place a blue pebble on any node $v$ that has a red pebble.
 \item \textsc{load}: place a red pebble on any node $v$ that has a blue pebble.
 \item \textsc{compute}: if all the inputs of a non-source node $v$ have a red pebble on them, then also place a red pebble on $v$.
 \item \textsc{delete}: remove a red pebble from any node $v$.
\end{enumerate}

A pebbling consists of a sequence of rule applications such that finally, all the sink nodes of $G$ (i.e.\ the outputs of the computation) have a blue pebble. Furthermore, we require that at any point during the pebbling, the number of red pebbles placed on $G$ is at most $r$, i.e.\ the fast memory capacity is not exceeded. The cost of a pebbling strategy is defined as the total number of save and load operations executed during the pebbling. Intuitively, the compute and delete steps are considered free; indeed, these are often significantly cheaper in practice.

One fundamental property of RBP is that it requires all inputs of an operation to be collected into fast memory at the same time. Specifically, if a node $v$ has $d$ incoming edges in our DAG, then any valid pebbling in RBP requires at least $(d\!+\!1)$ red pebbles simultaneously to compute this node: $d$ pebbles on the in-neighbors of $v$, and one more on $v$. Hence a valid pebbling only exists if $r \! \geq \! (\Delta_{in}\!+\!1)$.

In contrast to this, in many practical cases, the subtasks in our computation are associative and commutative operators, e.g.\ the addition of scalars or matrices, applied on a larger set of inputs. In this case, the separate inputs could be added to the aggregated value in any specific order, or if desired, we could even decide to first aggregate e.g.\ only half of the inputs of a node $v$, store the temporary result in (fast or slow) memory, then do some computations in another part of the DAG, and finally continue processing the remaining inputs of $v$.

The idea of incorporating such partial computation steps into RBP was recently raised by Sobczyk~\cite{partial}. Allowing partial computations not only enables us to analyze I/O costs when $r  < (\Delta_{in}\!+\!1)$, but it can also drastically reduce the number of required I/O steps for DAGs with $r \geq (\Delta_{in}\!+\!1)$. Thus, while this extension restricts the scope of RBP to computational DAGs with specific kind of operators, it provides a much more accurate model of the actual I/O cost of the computation in these domains.

In this paper, we extend the red-blue pebble game with such partial computations. We begin by discussing the fundamental properties of this partial-computing red-blue pebble game (\emph{PRBP}), and showing some example DAGs where allowing partial compute steps indeed decreases the number of required I/O operations. We then analyze how the optimum cost on a given DAG differs in RBP and PRBP in general. In particular, we show that partial computations can decrease the required I/O cost by up to a linear factor, but in a general DAG, it is NP-hard to decide whether they reduce the optimum cost at all.

We then focus on developing lower bounds on the I/O cost of specific computations in PRBP. We first show in a counterexample that lower bounds derived from $S$-partitions, which is a standard tool for RBP, do not directly carry over to PRBP. Instead, we discuss two different ways to modify the $S$-partition concept so that the resulting bounds also apply to PRBP. We then use these newly developed tools to derive I/O lower bounds on some crucial computational tasks, such as fast Fourier transform, matrix multiplication, or self-attention in neural networks. Since the optimal solutions to these problems have been studied extensively, and their I/O complexity matches known lower bounds from RBP, it is not surprising that the obtained PRBP lower bounds are identical to those in RBP, i.e.\ partial computations do not improve upon the state-of-the-art in these DAGs. Nonetheless, our formal proof of this fact is still a crucial step towards understanding the properties of these problems, and I/O complexity in general.

Finally, by adapting a proof construction from RBP to PRBP, we show that the best pebbling strategy in PRBP is still NP-hard to approximate to an $n^{1-\varepsilon}$ factor for any $\varepsilon>0$.

\section{Related Work}

The concept of pebble games on graphs has been widely used to capture various different aspects of computing, such as time-memory trade-offs, non-determinism or reversibility~\cite{cook1973observation, otherpebbling2, otherpebbling3}.

The red-blue pebble game was introduced by Hong and Kung to analyze the I/O complexity of a computational task~\cite{RBpebbling1}. Their seminal paper also presents $S$-partitions, a vital tool to derive corresponding I/O lower bounds on specific structured DAGs. These $S$-partitions were used to establish lower bounds for various computations both in~\cite{RBpebbling1} and in several following works \cite{elango2015characterizing, ranjan2012upper}. Other papers have used different methods to derive I/O lower bounds on concrete computations~\cite{jain2020spectral, RBpebbling5, scott2015matrix}. Recently, there is also a growing interest in using RBP to derive bounds on the I/O cost of crucial operations in modern neural networks~\cite{ml1, ml2}.

Besides lower bounds, another natural direction of work studies the computational complexity of finding the best pebbling strategy in RBP. This includes NP-hardness proofs for different variants of the model~\cite{RBpebbling2, RBpebbling3, RBpebbling8} and inapproximability results~\cite{mpp, RBpebbling3}, as well as approximation algorithms for special cases~\cite{RBpebbling5, RBpebbling7}.

There are also several papers that discuss the idea of generalizing the RBP model to multiple processors or to deeper memory hierarchies~\cite{mpp,RBpebbling7, savage1995extending,savage2008unified,elango2014characterizing}. However, these generalized settings are typically very complex, and thus the theoretical results in them are limited to concrete DAGs or special cases.

A recent preprint~\cite{partial} raises the idea of extending RBP with partial computations. Our work is inspired by this approach, but rather different both in terms of the model and the direction of the results. In particular, the model of~\cite{partial} assumes that all nodes have an initialization value that always needs to be loaded from slow memory, motivated by e.g.\ evaluating polynomials. This incurs an extra I/O step for each sink non-sink node, and hence results in a significantly different (often much higher) I/O cost for most DAGs than either the original RBP or our model. Due to this, most claims in our paper do not carry over to the setting of~\cite{partial}. In contrast to this, our PRBP model is a direct extension of RBP, where \emph{any} RBP pebbling strategy translates to a partial-computing PRBP solution of the same I/O cost.

In terms of results,~\cite{partial} mainly focuses on DAGs that are not even computable in RBP due to $r \! < \! (\Delta_{in} \! + \! 1)$, and shows NP-hardness and approximation methods for the edge case when the cache size is only $r\!=\!2$. In contrast, we focus on understanding the properties of PRBP in general, and its relation to RBP, for example, how the optimal I/O costs differ between the two models, and which general results from RBP carry over to PRBP.

Regarding the actual model definition, we also change the terminology from~\cite{partial} slightly to align it with the original RBP terminology. Specifically, we use ``light red'' and ``dark red'' pebbles to denote the two different states of values kept in cache, to signify that these all correspond to red pebbles in the original RBP. Furthermore, the incoming edges of a node that are already aggregated into the final value are not deleted, but rather marked in our case; this indicates that the structure of the computation remains unchanged, and also allows to extend the setting to re-computations.

\section{Pebbling With Partial Computations}

The partial-computing red-blue pebble game (PRBP) is defined as follows. Similarly to the original RBP, we have red and blue pebbles, which indicate that the output value of a node $v$ is stored in fast memory and slow memory, respectively. However, in PRBP, we now have two different kinds of red pebbles: a \emph{light red} pebble on node $v$ indicates that the current value of $v$ is also up-to-date in slow memory, whereas a \emph{dark red} pebble indicates that the newest value of $v$ is only present in fast memory, and thus, if ever needed in the future, it has to be first saved to slow memory before deletion. The total number of light and dark red pebbles on the DAG is again bounded at any point by the fast memory size $r$. The role of blue pebbles remains unchanged: they still represent values in slow memory, and their number is unlimited.

At any point during a PRBP pebbling, a node $v$ will have one of the following subsets of pebbles on it:
\begin{itemize}[leftmargin=1.5em, topsep=4pt, itemsep=1.5pt]
 \item no pebble, if its value is not stored anywhere,
 \item a blue pebble, if its value is only present in slow memory,
 \item a blue and a light red pebble, if its current value is present in both fast and slow memory,
 \item only a dark red pebble, if the value of $v$ has been updated since the last I/O operation on $v$, and hence the new value is only available in fast memory currently.
\end{itemize}

Since the computation of a node $v$ can now happen in multiple iterations, we also \emph{mark} the incoming edges of $v$ for the inputs that have already been aggregated into the current value of $v$. The final value of $v$ only becomes available when all of $v$'s incoming edges are marked; it is only after this point that the computation of $v$ is finished. We can then start using $v$ as an input for its successor operations, i.e.\ start marking the output edges of $v$.

The concrete transition rules for PRBP are as follows:
\begin{enumerate}[label=\arabic*), leftmargin=1.5em, topsep=4pt, itemsep=1.5pt]
 \item \textsc{save}: replace a dark red pebble on any node $v$ by a blue and a light red pebble.
 \item \textsc{load}: place a light red pebble on any node $v$ that has a blue pebble.
 \item \textsc{partial compute}: consider an unmarked edge $(u,v)$ such that: (i) all incoming edges of $u$ are marked, (ii) there is a (light or dark) red pebble on $u$, and (iii) there is either a (light or dark) red pebble on $v$, or no pebble on $v$ at all. Then replace all pebbles on $v$ by a dark red pebble, and mark the edge $(u,v)$.
 \item \textsc{delete}: remove a light red pebble from any node $v$, or a dark red pebble from a node $v$ that has all of its output edges marked.
\end{enumerate}

The save and load rules naturally implement the transitions between the states discussed above. The partial compute rule aggregates another input into the value of $v$; the conditions ensure that (i) $u$ is already fully computed, (ii) $u$ is in fast memory, and (iii) either the current value of $v$ is in fast memory, or we are just beginning to compute $v$. Finally, the delete rule ensures that values from fast memory can only be erased if they are either up-to-date in slow memory, or not needed anymore. This prohibits us from marking some incoming edges of $v$, deleting $v$'s red pebble, and then later getting a red pebble on $v$ simply by marking another edge; if we want to remove the red pebble from $v$ before finishing its computation, we need to save it to slow memory first.

Similarly to RBP, the initial state only has blue pebbles on the source nodes, and all edges are unmarked. In the terminal state, we now not only require that every sink has a blue pebble, but also that every edge of the DAG is marked; this indicates that all the computations have indeed been fully executed. At any point, the number of red pebbles on the DAG can be at most $r$. As in RBP, the cost of a pebbling is defined as the total number of save and load operations.

Note that in contrast to RBP, this game allows a valid pebbling for any DAG with as few as $r \! = \! 2$ red pebbles, by considering a topological order of $G$ and marking each incoming edge of the next node, with the inputs loaded and outputs saved when necessary.

In our paper, we consider the \emph{one-shot} version of RBP and PRBP, which means that during a pebbling, the compute rule can be applied only once for every node $v \! \in \! V$ in RBP, and only once for every edge $(u,v)\! \in \! E$ in PRBP. One-shot RBP is a prominent version in the literature~\cite{RBpebbling3, RBpebbling7, elango2014characterizing, elango2015characterizing}; we focus on this variant because allowing so-called re-computation steps in PRBP would raise further modeling questions. In Section~\ref{sec:models}, we briefly discuss other variants of RBP, and how the PRBP model extends to these settings.

For a given DAG $G$ and choice of $r$, we will use $\texttt{OPT}_{RBP}$ and $\texttt{OPT}_{PRBP}$ to denote the cost of the optimal pebbling strategy in RBP and PRBP, respectively. Regarding notation, we also commonly use $[d]$ to denote the set of integers $\{1, ..., d\}$.

\section{Fundamental properties of PRBP} \label{sec:io_relate}

One fundamental question about this extended model is how the partial compute steps affect the optimal I/O cost of a specific computation. Note that we only compare $\texttt{OPT}_{RBP}$ and $\texttt{OPT}_{PRBP}$ for cases when $(\Delta_{in}+1) \leq r$, i.e.\ a valid pebbling exists in both RBP and PRBP. 

First, one can observe that any pebbling strategy in RBP can easily be converted into a PRBP strategy of the same cost, by simply taking any compute step in RBP, and separating it into at most $\Delta_{in}$ consecutive partial compute steps in PRBP.

\begin{proposition}
For any DAG and any $r \geq (\Delta_{in}\!+\!1)$, we have \[\texttt{OPT}_{PRBP} \, \leq \, \texttt{OPT}_{RBP} \, . \]
\end{proposition}

Furthermore, note that both in RBP and PRBP, any valid pebbling strategy needs to load every source node at least once, and needs to save every sink node at least once. If the total number of source and sink nodes in $G$ is $t$, then we refer to this value $t$ as the \emph{trivial cost}, and we have $\texttt{OPT}_{RBP} \geq t$ and $\texttt{OPT}_{PRBP} \geq t$.

\subsection{Motivating examples}

There are numerous DAGs and choices of $r$ with $\texttt{OPT}_{PRBP}  < \texttt{OPT}_{RBP}$, i.e.\ where allowing partial computations indeed decreases the optimal I/O cost. We first discuss some simple examples.

\begin{proposition} \label{prop:example}
For the small example DAG in Figure~\ref{fig:decreased_opt} with $r=4$, we have $\texttt{OPT}_{RBP} = 3$ but $\texttt{OPT}_{PRBP} = 2$.
\end{proposition}

\begin{proof}
Consider the DAG in Figure~\ref{fig:decreased_opt}, with $u_0$, $v_0$ and the dashed edges also included. Note that both in RBP and PRBP, $u_0$ needs to be loaded in the beginning and $v_0$ needs to be saved in the end, so the I/O cost is at least $2$.

In PRBP, the DAG can be pebbled without any further I/O steps. The key to this is the partial computation of the node $w_3$. After loading $u_0$, we first compute $u_1$ and $u_2$ (placing dark red pebbles on them), and then delete the light red pebble from $u_0$. We then compute first along the edge $(u_1,w_1)$, and then afterwards along $(w_1,w_3)$, partially computing $w_3$. This puts a dark red pebble on $w_1$ and then $w_3$. The pebble from $w_1$ can be deleted immediately afterwards. Then similarly, we can compute along the edges $(u_1,w_2)$ and then $(w_2,w_3)$, placing a dark red on $w_2$ and updating the dark red on $w_3$, and then the pebble from $w_2$ can be deleted. After this, we can compute $w_4$ (mark both its incoming edges) and delete the reds from $u_1$ and $w_3$. Finally, we can compute $v_1$ and $v_2$ (marking their incoming edges in any order), delete the red pebbles from $u_1$ and $u_2$, compute $v_0$, and save the value of $v_0$.

In contrast to this, in RBP, we cannot compute the DAG without a further I/O step. In particular, computing $w_3$ now requires a red pebble on $w_1$, $w_2$ and $w_3$ at the same time. Due to this, when $w_3$ is computed, we can only have a red pebble on one of the nodes from $\{ u_0, u_1, u_2\}$. If this red pebble is on $u_0$ or $u_2$, then we need to later load $u_1$ from slow memory to compute $w_4$. Otherwise, assume that the red pebble is on $u_1$. Then if $u_2$ has a blue pebble, we need to load $u_2$ for computing $v_1$; if $u_2$ has no blue pebble, then we need to load $u_0$ again to compute $u_2$. In either case, we need another I/O step, hence $\texttt{OPT}_{RBP} \geq 3$. However, if $u_2$ is computed after $w_4$ by loading $u_0$ again, then we can indeed finish with only $3$ I/O steps.

For completeness, we list the concrete steps of the optimal pebbling strategies in both PRBP and RBP in Appendix~\ref{app:example}. 
\end{proof}

\begin{figure}
    \centering
    \begin{tikzpicture}

    \node[anchor=center] at (-11pt,-3pt) {\large $u_2$};
    \node[anchor=center] at (-11pt,43pt) {\large $u_1$};
    \node[anchor=center] at (150.5pt,-3pt) {\large $v_2$};
    \node[anchor=center] at (150.5pt,43pt) {\large $v_1$};
    \node[anchor=center] at (35pt,50.5pt) {\large $w_2$};
    \node[anchor=center] at (35pt,88pt) {\large $w_1$};
    \node[anchor=center] at (65pt,68pt) {\large $w_3$};
    \node[anchor=center] at (95pt,48pt) {\large $w_4$};

    \node[anchor=center] at (-44pt,12pt) {\large $u_0$};
    \node[anchor=center] at (184pt,12pt) {\large $v_0$};

    \begin{scope}[thick, arrows=-stealth]
    \draw (0pt,0pt) -- (136pt,0pt);
    \draw (0pt,0pt) -- (137.5pt,36.5pt);
    \draw (0pt,40pt) -- (86pt,40pt);
    \draw (90pt,40pt) -- (136pt,40pt);
    \draw (90pt,40pt) -- (137.5pt,3pt);
    \draw (0pt,40pt) -- (26pt,59pt);
    \draw (0pt,40pt) -- (27pt,78pt);
    \draw (30pt,60pt) -- (56pt,60pt);
    \draw (30pt,80pt) -- (57.25pt,63pt);
    \draw (60pt,60pt) -- (87.25pt,43pt);

    \draw[densely dashed] (-40pt,20pt) -- (-4pt,1pt);
    \draw[densely dashed] (-40pt,20pt) -- (-4pt,39pt);
    \draw[densely dashed] (140pt,0pt) -- (177pt,18pt);
    \draw[densely dashed] (140pt,40pt) -- (177pt,22pt);
    \end{scope}

    \draw[black, fill=white] (0pt,0pt) circle (1.0ex);
    \draw[black, fill=white] (0pt,40pt) circle (1.0ex);

    \draw[black, fill=white] (30pt,60pt) circle (1.0ex);
    \draw[black, fill=white] (30pt,80pt) circle (1.0ex);
    \draw[black, fill=white] (60pt,60pt) circle (1.0ex);
    \draw[black, fill=white] (90pt,40pt) circle (1.0ex);

    \draw[black, fill=white] (140pt,0pt) circle (1.0ex);
    \draw[black, fill=white] (140pt,40pt) circle (1.0ex);

    \draw[black, fill=white] (-40pt,20pt) circle (1.0ex);
    \draw[black, fill=white] (180pt,20pt) circle (1.0ex);

\end{tikzpicture}
    \caption{Example DAG for $\texttt{OPT}_{PRBP} < \texttt{OPT}_{RBP}$, with $r=4$. For Proposition~\ref{prop:example}, $u_0$, $v_0$ and the dashed edges are part of the DAG; in this case, we have $\texttt{OPT}_{PRBP} = 2$ but $\texttt{OPT}_{RBP} = 3$. For Proposition~\ref{prop:reduced}, we disregard $u_0$, $v_0$ and the dashed edges, and concatenate several copies of the remaining gadget.}
    \label{fig:decreased_opt}
\end{figure}
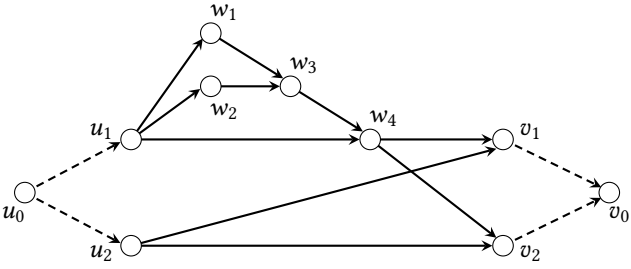

Towards a more practical example, consider a fundamental operation in linear algebra, namely, the multiplication $A \! \cdot \! x=y$ of an $m \times m$ matrix $A$ with an $m \times 1$ vector $x$. The corresponding computational DAG consists of $m^2\!+m$ source nodes, $m^2$ intermediate nodes with in-degree $2$, and $m$ sink nodes of in-degree $m$.

\begin{proposition} \label{prop:mat_vec}
If $\,m \geq 3$ and $\,m\!+\!3 \leq r \leq 2m$, then for the matrix-vector multiplication above, we have $\texttt{OPT}_{PRBP}  < \texttt{OPT}_{RBP}$.
\end{proposition}

\begin{proof}
In PRBP, we can devise a strategy that keeps the $m$ partially computed values of the output $y$ in fast memory, and only uses $3$ further red pebbles besides this. In particular, we can iterate through the entries $x_i$ of the input vector, first loading $x_i$ into fast memory, and then for each row of the matrix, we load the corresponding $A_{j,i}$, multiply it with $x_i$, and add this to $y_j$. Note that $A_{j,i}$ can be deleted after computing $A_{j,i}\!\cdot\! x_i$, and then $A_{j,i}\!\cdot\!x_i$ can be deleted after updating $y_j$, so this only requires two red pebbles besides the one on $x_i$. Afterwards, $x_i$ can also be deleted. This strategy comes with only the trivial I/O cost of $\texttt{OPT}_{PRBP} = m^2+2m$.

In RBP, we can show that at least one non-trivial I/O step always happens between the computation of any two consecutively computed sink nodes $y_{j_1}$ and $y_{j_2}$. Note that computing $y_{j_1}$ requires $(m+1)$ red pebbles simultaneously, and thus we can have at most $(m-1)$ red pebbles in the remaining nodes. As such, there is an index $i \in [m]$ such that neither the source node $x_i$ nor the intermediate node $A_{i, j_2}$ has a red pebble. If $A_{i, j_2}$ already has a blue pebble, then this means we have to load it before computing $y_{j_2}$. Otherwise, if $A_{i, j_2}$ has no pebble, then we need to load $x_i$ to compute it; note that $x_i$ has been loaded before to compute $A_{i, j_1}$, so this is also a non-trivial I/O. Altogether, this gives $\texttt{OPT}_{RBP} \geq m^2\!+\!3m-1$.
\end{proof}

\subsection{Frequently used gadgets}

We next analyze the behavior of some well-structured graphs in PRBP, illustrated in Figure~\ref{fig:gadgets}. These often appear as building blocks in larger computations, or as gadgets in proof constructions.

\subsubsection{Zipper gadget}

One popular gadget in RBP is the zipper gadget from~\cite{RBpebbling3, mpp}. This consists of two groups of $d$ source nodes, and a long chain where each node has incoming edges from the previous chain node, and one of the two source groups alternatingly.

The key property of this gadget in RBP is that if $r=d+2$, then for each next node in the chain, we need to move $d$ red pebbles from one group to the other. This results in $d$ load operations for each chain node. As we increase $r$, this I/O cost induced by each step gradually decreases, and disappears completely at $r=2d+2$.

In contrast, in PRBP with $r=d+2$, once we have red pebbles in one of the groups, we can compute the corresponding partial value for each node of the chain, and save these into slow memory. Then we can keep the red pebbles in the other group for the traversal of the whole chain, and begin at each chain node by loading the previously computed partial value. This only results in a cost of $2$ for each chain node, which is a lower total cost when $d \geq 3$.

\begin{proposition}
If $d \geq 3$ and $r=d+2$, then $\texttt{OPT}_{PRBP}  < \texttt{OPT}_{RBP}$ in the zipper gadget.
\end{proposition}

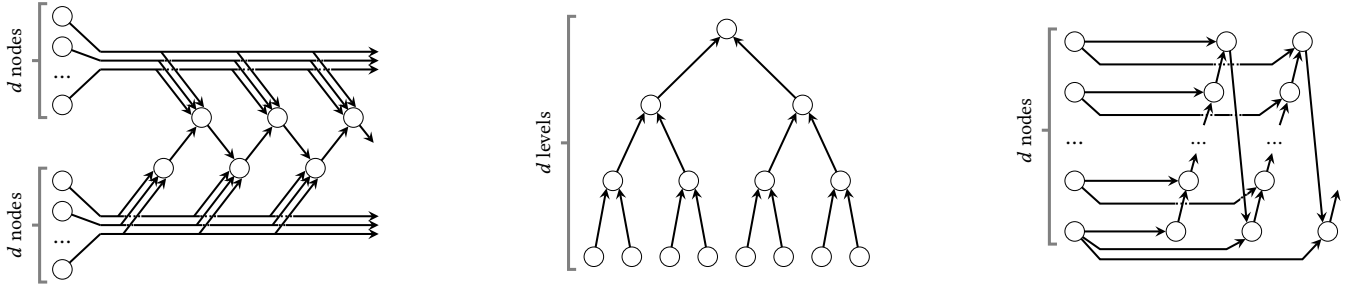
\begin{figure*}
    \centering
    \resizebox{1.0\textwidth}{!}{\begin{tikzpicture}

    \begin{scope}[thick]
    \draw (0pt,0pt) -- (15pt,14pt);
    \draw (0pt,23pt) -- (15pt,17.5pt);
    \draw (0pt,35pt) -- (15pt,21pt);

    \draw (0pt,65pt) -- (15pt,79pt);
    \draw (0pt,88pt) -- (15pt,82.5pt);
    \draw (0pt,100pt) -- (15pt,86pt);
    \end{scope}

    \begin{scope}[thick, arrows=-stealth]
    \draw (15pt,14pt) -- (125pt,14pt);
    \draw (15pt,17.5pt) -- (125pt,17.5pt);
    \draw (15pt,21pt) -- (125pt,21pt);

    \draw (15pt,79pt) -- (125pt,79pt);
    \draw (15pt,82.5pt) -- (125pt,82.5pt);
    \draw (15pt,86pt) -- (125pt,86pt);
    \end{scope}

    \draw[white, fill=white] (25.5pt,16.5pt) rectangle (28pt,18.5pt);
    \draw[white, fill=white] (55.5pt,16.5pt) rectangle (58pt,18.5pt);
    \draw[white, fill=white] (85.5pt,16.5pt) rectangle (88pt,18.5pt);

    \draw[white, fill=white] (40.5pt,81.5pt) rectangle (43pt,83.5pt);
    \draw[white, fill=white] (70.5pt,81.5pt) rectangle (73pt,83.5pt);
    \draw[white, fill=white] (100.5pt,81.5pt) rectangle (103pt,83.5pt);

    \draw[white, fill=white] (24.5pt,20pt) rectangle (31pt,22pt);
    \draw[white, fill=white] (54.5pt,20pt) rectangle (61pt,22pt);
    \draw[white, fill=white] (84.5pt,20pt) rectangle (91pt,22pt);

    \draw[white, fill=white] (39.5pt,78pt) rectangle (46pt,80pt);
    \draw[white, fill=white] (69.5pt,78pt) rectangle (76pt,80pt);
    \draw[white, fill=white] (99.5pt,78pt) rectangle (106pt,80pt);

    \begin{scope}[thick, densely dotted]
    \draw (25.5pt,17.5pt) -- (28pt,17.5pt);
    \draw (24.5pt,21pt) -- (31pt,21pt);
    \draw (55.5pt,17.5pt) -- (58pt,17.5pt);
    \draw (54.5pt,21pt) -- (61pt,21pt);
    \draw (85.5pt,17.5pt) -- (88pt,17.5pt);
    \draw (84.5pt,21pt) -- (91pt,21pt);

    \draw (40.5pt,82.5pt) -- (43pt,82.5pt);
    \draw (39.5pt,79pt) -- (46pt,79pt);
    \draw (70.5pt,82.5pt) -- (73pt,82.5pt);
    \draw (69.5pt,79pt) -- (76pt,79pt);
    \draw (100.5pt,82.5pt) -- (103pt,82.5pt);
    \draw (99.5pt,79pt) -- (106pt,79pt);
    \end{scope}

    \begin{scope}[thick, arrows=-stealth]

    \draw (24pt,14pt) -- (41pt,36pt);
    \draw (23pt,17.5pt) -- (38pt,37pt);
    \draw (22pt,21pt) -- (36pt,39pt);

    \draw (54pt,14pt) -- (71pt,36pt);
    \draw (53pt,17.5pt) -- (68pt,37pt);
    \draw (52pt,21pt) -- (66pt,39pt);

    \draw (84pt,14pt) -- (101pt,36pt);
    \draw (83pt,17.5pt) -- (98pt,37pt);
    \draw (82pt,21pt) -- (96pt,39pt);

    \draw (39pt,86pt) -- (56pt,64pt);
    \draw (38pt,82.5pt) -- (53pt,63pt);
    \draw (37pt,79pt) -- (51pt,61pt);

    \draw (69pt,86pt) -- (86pt,64pt);
    \draw (68pt,82.5pt) -- (83pt,63pt);
    \draw (67pt,79pt) -- (81pt,61pt);

    \draw (99pt,86pt) -- (116pt,64pt);
    \draw (98pt,82.5pt) -- (113pt,63pt);
    \draw (97pt,79pt) -- (111pt,61pt);

    \draw (40pt,40pt) -- (53pt,56.5pt);
    \draw (70pt,40pt) -- (83pt,56.5pt);
    \draw (100pt,40pt) -- (113pt,56.5pt);

    \draw (55pt,60pt) -- (68pt,43.5pt);
    \draw (85pt,60pt) -- (98pt,43.5pt);
    \draw (115pt,60pt) -- (123pt,50pt);
    \end{scope}

    \draw[black, fill=white] (0pt,0pt) circle (1.0ex);
    \node[anchor=center] at (0pt,11pt) {\large $...$};
    \draw[black, fill=white] (0pt,23pt) circle (1.0ex);
    \draw[black, fill=white] (0pt,35pt) circle (1.0ex);

    \draw[black, fill=white] (0pt,65pt) circle (1.0ex);
    \node[anchor=center] at (0pt,76pt) {\large $...$};
    \draw[black, fill=white] (0pt,88pt) circle (1.0ex);
    \draw[black, fill=white] (0pt,100pt) circle (1.0ex);

    \draw[black, fill=white] (40pt,40pt) circle (1.0ex);
    \draw[black, fill=white] (55pt,60pt) circle (1.0ex);
    \draw[black, fill=white] (70pt,40pt) circle (1.0ex);
    \draw[black, fill=white] (85pt,60pt) circle (1.0ex);
    \draw[black, fill=white] (100pt,40pt) circle (1.0ex);
    \draw[black, fill=white] (115pt,60pt) circle (1.0ex);

    \begin{scope}[very thick, gray]
    \draw (-6pt,-5pt) -- (-9pt,-5pt) -- (-9pt,40pt) -- (-6pt,40pt);
    \draw (-9pt,17.5pt) -- (-12pt,17.5pt);

    \draw (-6pt,60pt) -- (-9pt,60pt) -- (-9pt,105pt) -- (-6pt,105pt);
    \draw (-9pt,82.5pt) -- (-12pt,82.5pt);
    \end{scope}

    \node[anchor=center, rotate=90] at (-19pt,17.5pt) {\small $d$ nodes};
    \node[anchor=center, rotate=90] at (-19pt,82.5pt) {\small $d$ nodes};


    \begin{scope}[thick, arrows=-stealth]
    \draw (210pt,5pt) -- (215pt,32pt);
    \draw (225pt,5pt) -- (220pt,32pt);
    \draw (240pt,5pt) -- (245pt,32pt);
    \draw (255pt,5pt) -- (250pt,32pt);
    \draw (270pt,5pt) -- (275pt,32pt);
    \draw (285pt,5pt) -- (280pt,32pt);
    \draw (300pt,5pt) -- (305pt,32pt);
    \draw (315pt,5pt) -- (310pt,32pt);

    \draw (217.5pt,35pt) -- (230pt,62pt);
    \draw (247.5pt,35pt) -- (235pt,62pt);
    \draw (277.5pt,35pt) -- (290pt,62pt);
    \draw (307.5pt,35pt) -- (295pt,62pt);

    \draw (232.5pt,65pt) -- (260pt,92pt);
    \draw (292.5pt,65pt) -- (265pt,92pt);
    \end{scope}

    \draw[black, fill=white] (210pt,5pt) circle (1.0ex);
    \draw[black, fill=white] (225pt,5pt) circle (1.0ex);
    \draw[black, fill=white] (240pt,5pt) circle (1.0ex);
    \draw[black, fill=white] (255pt,5pt) circle (1.0ex);
    \draw[black, fill=white] (270pt,5pt) circle (1.0ex);
    \draw[black, fill=white] (285pt,5pt) circle (1.0ex);
    \draw[black, fill=white] (300pt,5pt) circle (1.0ex);
    \draw[black, fill=white] (315pt,5pt) circle (1.0ex);

    \draw[black, fill=white] (217.5pt,35pt) circle (1.0ex);
    \draw[black, fill=white] (247.5pt,35pt) circle (1.0ex);
    \draw[black, fill=white] (277.5pt,35pt) circle (1.0ex);
    \draw[black, fill=white] (307.5pt,35pt) circle (1.0ex);

    \draw[black, fill=white] (232.5pt,65pt) circle (1.0ex);
    \draw[black, fill=white] (292.5pt,65pt) circle (1.0ex);

    \draw[black, fill=white] (262.5pt,95pt) circle (1.0ex);

    \begin{scope}[very thick, gray]
    \draw (203pt,0pt) -- (200pt,0pt) -- (200pt,100pt) -- (203pt,100pt);
    \draw (200pt,50pt) -- (197pt,50pt);
    \end{scope}

    \node[anchor=center, rotate=90] at (190pt,50pt) {\small $d$ levels};


    \begin{scope}[thick, arrows=-stealth]
    \draw (400pt,35pt) -- (410pt,26pt) -- (463pt,26pt) -- (472pt,32.5pt);
    \draw (400pt,70pt) -- (410pt,61pt) -- (473pt,61pt) -- (482pt,67.5pt);
    \draw (400pt,90pt) -- (410pt,81pt) -- (478pt,81pt) -- (487pt,87.5pt);
    \end{scope}

    \draw[white, fill=white] (440pt,25pt) rectangle (445pt,27pt);
    \draw[white, fill=white] (450pt,60pt) rectangle (455pt,62pt);
    \draw[white, fill=white] (455pt,80pt) rectangle (465pt,82pt);
    
    \draw[white, fill=white] (461.5pt,60pt) rectangle (466pt,62pt);
    \draw[white, fill=white] (465.7pt,26pt) rectangle (468pt,30.5pt);

    \begin{scope}[thick, densely dotted]
    \draw (440pt,26pt) -- (445pt,26pt);
    \draw (450pt,61pt) -- (455pt,61pt);
    \draw (455pt,81pt) -- (465pt,81pt);
    \draw (461.5pt,61pt) -- (466pt,61pt);
    \draw (465.8pt,28.2pt) -- (468pt,29.7pt);
    \end{scope}

    \begin{scope}[thick, arrows=-stealth]
    \draw (440pt,15pt) -- (444pt,31.5pt);
    \draw (455pt,70pt) -- (459pt,86.5pt);
    \draw (461pt,90pt) -- (468pt,18.5pt);

    \draw (470pt,15pt) -- (474pt,31.5pt);
    \draw (485pt,70pt) -- (489pt,86.5pt);
    \draw (491pt,90pt) -- (498pt,18.5pt);

    \draw (500pt,15pt) -- (504pt,31.5pt);

    \draw (445pt,35pt) -- (448pt,46pt);
    \draw (451pt,57pt) -- (454pt,66pt);

    \draw (475pt,35pt) -- (478pt,46pt);
    \draw (481pt,57pt) -- (484pt,66pt);

    \draw (400pt,15pt) -- (436pt,15pt);
    \draw (400pt,35pt) -- (441pt,35pt);
    \draw (400pt,70pt) -- (451pt,70pt);
    \draw (400pt,90pt) -- (456pt,90pt);

    \draw (400pt,15pt) -- (410pt,8pt) -- (460pt,8pt) -- (467.5pt,12pt);
    \draw (400pt,15pt) -- (410pt,4pt) -- (490pt,4pt) -- (497.5pt,12pt);

    \end{scope}

    \node[anchor=center] at (449pt,50pt) {$...$};
    \node[anchor=center] at (479pt,50pt) {$...$};

    \draw[black, fill=white] (400pt,15pt) circle (1.0ex);
    \draw[black, fill=white] (400pt,35pt) circle (1.0ex);
    \node[anchor=center] at (400pt,50pt) {\large $...$};
    \draw[black, fill=white] (400pt,70pt) circle (1.0ex);
    \draw[black, fill=white] (400pt,90pt) circle (1.0ex);

    \draw[black, fill=white] (440pt,15pt) circle (1.0ex);
    \draw[black, fill=white] (445pt,35pt) circle (1.0ex);
    \draw[black, fill=white] (455pt,70pt) circle (1.0ex);
    \draw[black, fill=white] (460pt,90pt) circle (1.0ex);

    \draw[black, fill=white] (470pt,15pt) circle (1.0ex);
    \draw[black, fill=white] (475pt,35pt) circle (1.0ex);
    \draw[black, fill=white] (485pt,70pt) circle (1.0ex);
    \draw[black, fill=white] (490pt,90pt) circle (1.0ex);

    \draw[black, fill=white] (500pt,15pt) circle (1.0ex);

    \begin{scope}[very thick, gray]
    \draw (393pt,10pt) -- (390pt,10pt) -- (390pt,95pt) -- (393pt,95pt);
    \draw (390pt,55pt) -- (387pt,55pt);
    \end{scope}

    \node[anchor=center, rotate=90] at (380pt,55pt) {\small $d$ nodes};

\end{tikzpicture}}
    \caption{Illustration of some frequently used substructures and gadgets: the zipper gadget of~\cite{RBpebbling3, mpp} with $2 \cdot d$ source nodes and an alternating chain (left), a binary tree on $2^d$ leaves with all edges pointing towards the root (middle), and a pebble collection gadget from~\cite{RBpebbling3} with $d$ source nodes and a chain that periodically uses the different sources as an input (right).}
    \label{fig:gadgets}
\end{figure*}

\subsubsection{Binary trees}

Another structure that often occurs in the context of RBP is a rooted binary (or $k$-ary) tree, i.e.\ a tree of depth $d$ with every non-leaf node having exactly $2$ (or exactly $k$) distinct in-neighbors. These trees often show up as smaller parts of concrete computational tasks~\cite{ml1}, or as examples to illustrate the properties of different model variants~\cite{mpp}.

In terms of pebbling, the most interesting behavior for binary trees is observed when $r=3$. Due to their very regular structure, the optimal pebbling strategy is easy to find in these DAGs. Note that we always have to load the $2^d$ leaves and save the root node, so the trivial I/O cost is $2^d +1$. Apart from this, in RBP, for each node $v$ that is not in the bottom two levels (i.e., not a leaf or an out-neighbor of a leaf), we will always need to save (and later reload) one of the two in-neighbors of $v$, in order to free up a red pebble and also compute the other subtree of $v$. This way, the number of non-trivial I/O steps in the optimal RBP pebbling is
\[ 2 \cdot \left( 1 + ... + 2^{d-2} \right) = 2^d - 2 \, . \]

In contrast, in PRBP, the bottom $3$ levels can be computed without any I/O cost due to partial computations. However, every node $v$ above these levels once again requires $2$ I/O steps, regardless of whether save the partially computed value of $v$ and load it later, or if we save one of $v$'s in-neighbors and load it later. Altogether, the number of non-trivial I/O steps here is
\[ 2 \cdot \left( 1 + ... + 2^{d-3} \right) = 2^{d-1} - 2 \, . \]

In Appendix~\ref{app:k-ary}, we discuss the pebbling strategy in more detail for a small binary tree, and also outline the generalization of these observations to $k$-ary trees.

\begin{proposition} \label{prop:bintree}
In binary trees of depth $d \geq 3$ with $r=3$, we again have $\texttt{OPT}_{PRBP}  < \texttt{OPT}_{RBP}$.
\end{proposition}

\subsubsection{Pebble-collecting gadgets}

When proving properties of RBP in a construction, there is often need for a gadget which ensures that we need to place all red pebbles on it simultaneously, or otherwise we incur some I/O cost. One example of such gadgets is the pyramid gadget~\cite{RBpebbling2,ranjan2012upper}, which is similar to a binary tree. A more sophisticated version of the gadget is the one from~\cite{RBpebbling3}, shown in the right side of Figure~\ref{fig:gadgets}. This gadget consist of $d$ source nodes, and a long chain of $\ell$ nodes, with the $i$-th node in the chain having incoming edges from the previous chain node and the $i$-th source node modulo $d$.

This gadget ensures in RBP that if we use $d\!+\!2$ red pebbles, then we can pebble it without any I/O (apart from the trivial cost). However, if there is no point in the pebbling where we have $d\!+\!2$ red pebbles on the gadget simultaneously, then the red pebbles need to be repeatedly reloaded to compute the chain, incurring a cost of at least $\ell/d$.

As a tool for future constructions, we show that the same gadget ensures a similar behavior in PRBP. Note that if we use $d+2$ red pebbles, then the gadget can also be pebbled in PRBP with only the trivial I/O cost. Otherwise, the pebbling again becomes rather costly.

\begin{proposition} \label{prop:coll}

If a PRBP strategy never has at least $d\!+\!2$ red pebbles simultaneously on this gadget, then the cost of the strategy is at least $\frac{\ell}{2d}$.

\end{proposition}

\begin{proof}
Assume that we never have $d+2$ red pebbles on the gadget simultaneously. Let us split the chain into segments of length $2d$; we show that each segment incurs at least one I/O step. Let $v_1$, ..., $v_{2d}$ denote the current chain segment. Consider the time when we mark the edge $(v_1, v_2)$, and thus we have red pebbles on both $v_1$ and $v_2$. By assumption, we have at most $d_{\!}-_{\!}1$ other red pebbles in the gadget. As such, if we consider the $d$ source nodes and the $d$ chain nodes $v_{d+1}$, ..., $v_{2d}$, there is at least one index $i \!\in\! [d]$ such that neither $u_i$ nor $v_{d+i}$ has a red pebble. Then the computation of $v_{d+i}$ will require a load operation, on $u_i$ if the edge $(u_i, v_{d+i})$ is unmarked, and on $v_{d+i}$ if it is marked already. This adds up to a cost of $\frac{\ell}{2d}$ over the whole chain.
\end{proof}

\subsection{On the ratio between $\texttt{OPT}_{RBP}$ and $\texttt{OPT}_{PRBP}$}

The previous examples also raise a natural question: how much can the partial compute steps reduce the optimum cost? We show below that this cost decrease can be very large in general, up to a linear factor in $n$. Furthermore, this can also happen in DAGs which, unlike the zipper gadget before, have very small degrees: we show an example for a linear cost decrease in a construction with $\Delta_{in} = 2$ and $\Delta_{out} = 3$.

\begin{proposition} \label{prop:reduced}
There exists a DAG construction with $\Delta_{in} = 2$ and $\Delta_{out} = 3$ such that we have $\texttt{OPT}_{RBP} = \Theta(n)$ but $\texttt{OPT}_{PRBP} = 2$.
\end{proposition}

\begin{proof}
The proof uses many copies of the gadget shown before in Figure~\ref{fig:decreased_opt}, concatenated in a serial fashion. That is, let us remove the nodes $u_0$, $v_0$ and the dashed edges from the DAG in the figure. Then consider many copies of the resulting gadget, always merging node $v_1$ of the $i$-th gadget with node $u_1$ of the $(i\!+\!1)$-th gadget, and node $v_2$ of the $i$-th gadget with node $u_2$ of the $(i\!+\!1)$-th gadget. In the end, we reinsert the source $u_0$ and add two edges towards $u_1$ and $u_2$ in the first gadget, and reinsert $v_0$, and add two edges from $v_1$ and $v_2$ in the last gadget. Note that the gadget consists of only $8$ nodes, so the resulting DAG contains $\Theta(n)$ copies of the gadget. Assume again that $r=4$.

As before, the trivial cost of $2$ is unavoidable in both RBP and PRBP. However, in PRBP, no further cost is required. If we have a dark red pebble on both $u_1$ and $u_2$, we can compute the gadget without further I/O steps as discussed before, placing dark red pebbles on $v_1$ and $v_2$ in the end; then we can continue with the next gadget similarly. Thus altogether, $\texttt{OPT}_{PRBP} = 2$.

On the other hand, in RBP, consider the subsequence of the pebbling between first placing a red pebble on $u_1$ and first placing a red pebble on $v_1$; we show that at least one I/O step happens in this subsequence. Once again, when $w_3$ is computed, we need to have red pebbles on all of $\{w_1, w_2, w_3\}$ simultaneously, leaving only one red pebble. If this last red pebble is not on $u_1$, then $u_1$ needs to be loaded later to compute $w_4$. If the last red pebble is on $u_1$, then there are two cases like before: if $u_2$ has a blue pebble, then we need to load $u_2$ later for $v_1$, and if $u_2$ has no pebble, then we need to load the input(s) of $u_2$ to compute it. Thus every gadget incurs a cost of at least $1$; this adds up to a total of $\Theta(n)$.
\end{proof}

\subsection{Complexity of $\texttt{OPT}_{RBP} = \texttt{OPT}_{PRBP}$}

It also turns out that in general, it is hard to decide whether partial compute steps could decrease the optimal cost in a concrete DAG.

\begin{theorem} \label{th:nphard}
Given a DAG $G$ and an integer $r$, it is NP-hard to decide whether $\texttt{OPT}_{PRBP} < \texttt{OPT}_{RBP}$.
\end{theorem}

This is a more technical proof that requires a modified version of a DAG construction from previous works~\cite{RBpebbling3, mpp}. This allows for a reduction from a version of the maximum independent set problem. Recall that given an undirected graph $G_0=(V_0, E_0)$, an independent set is a subset of nodes $V' \subseteq V_0$ that induces no edges, and a maximum independent set is one that maximizes $|V'|$ over all possible independent sets in $G_0$.

\begin{definition}
Let \textsc{maxinset-vertex} denote the following problem: given an undirected graph $G_0$ and node $v_0$ in $G_0$, is there a maximum independent set in $G_0$ that contains $v_0$?
\end{definition}

Despite the extensive literature on clique and independent-set problems, we are not aware that this specific variant has been proven NP-hard, and hence we include a proof in Appendix~\ref{app:clique}.

\begin{lemma} \label{lem_inset}
The \textsc{maxinset-vertex} problem is NP-hard.
\end{lemma}

\renewcommand*{\proofname}{Proof sketch for Theorem~\ref{th:nphard}}
\begin{proof}
The main ingredient of the reduction is the pebble collection gadget of Proposition~\ref{prop:coll}, with a choice of $d=r-2$; we have shown that this works identically in PRBP. Given a \textsc{maxinset-vertex} problem, for each node $u_0$ in $G_0$, we create two such gadgets, i.e.\ two groups of $(r-2)$ source nodes, denoted $H_1(u_0)$ and $H_2(u_0)$, and their corresponding chains. Besides this, for some parameter $b$, we merge the $i$-th node of $H_1(u_0)$ and the $i$-th node of $H_2(u_0)$ into a single node for all $i \in [b]$. Finally, for each edge $(u_1, u_2)$ in $G_0$, we replace a node in $H_2(u_2)$ by a node in the middle of the chain of $H_1(u_1)$, and vice versa, we replace a node in $H_2(u_1)$ by a node in the middle of the chain of $H_1(u_2)$.

With $\ell$ large enough, the gadgets ensure that we need to have all the $r$ red pebbles on them simultaneously in any reasonable solution. In particular, if for every gadget we use all the red pebbles together to compute it, then the total cost in the whole DAG is less than $\frac{\ell}{2(r-2)}$ even in the worst case, whereas if we do not do this for any single gadgets, the cost becomes at least $\frac{\ell}{2(r-2)}$ according to Proposition~\ref{prop:coll}. Due to this, any reasonable pebbling strategy comes down to the order of `visiting' the gadgets, i.e.\ having all red pebbles on them. If $H_1(u_0)$ and $H_2(u_0)$ are visited consecutively, then this saves a cost of $b$ by not having to reload their merged source nodes again. However, the set of groups visited consecutively can only be an independent set in $G_0$, due to the dependencies introduced between groups. As such, in both RBP and PRBP, the optimal pebbling visits the gadget-pairs of a largest-possible independent set in $G_0$ consecutively. For more details on the main idea of the construction, see ~\cite{RBpebbling3, mpp}.

To extend this reduction idea to our current theorem, consider the input node $v_0$ from \textsc{maxinset-vertex}. We further take two sets $Z_1$, $Z_2$ of $3$ nodes from $H_1(v_0)$ and $H_2(v_0)$, respectively. We add an extra sink node $w$ with edges from all $6$ nodes in $Z_1$ and $Z_2$. Saving the final value of $w$ always incurs an extra trivial cost of $1$; we disregard this step in the analysis below.

If there is an optimal pebbling where $H_1(v_0)$ and $H_2(v_0)$ are visited consecutively (i.e.\ $v_0$ is contained in a maximum independent set), we can easily create a point in the pebbling where all nodes of $Z_1$ and $Z_2$ have a red pebble, by ensuring that the pebbles of $Z_1$ are the last ones from $H_1(v_0) \setminus H_2(v_0)$ to be deleted, and that the pebbles of $Z_2$ are the first ones from $H_2(v_0) \setminus H_1(v_0)$ to be placed. Assuming $r > b + 7$, we can easily compute $w$ at this point at no extra cost in both RBP and PRBP. On the other hand, if $H_1(v_0)$ and $H_2(v_0)$ are not consecutive in any reasonable pebbling, then no optimal visitation has pebbles on $Z_1$ and $Z_2$ simultaneously. In this case, in PRBP, we can partially compute $w$ from the inputs in $Z_1$, save it to slow memory, load it back when $Z_2$ has red pebbles, and finish computing it, at a further cost of only $2$. However, in RBP, we need to reload all $3$ nodes of e.g.\ $Z_1$ once again when we only have red pebbles on $Z_2$, and compute $w$ at this point, incurring further I/O cost of $3$. As such, we have $\texttt{OPT}_{PRBP} < \texttt{OPT}_{RBP}$ if and only if $H_1(v_0)$ and $H_2(v_0)$ cannot be consecutive in an optimal pebbling, i.e.\ when \textsc{maxinset-vertex} evaluates to false.

For more details on the construction, see Appendix~\ref{app:np}.
\end{proof}
\renewcommand*{\proofname}{Proof}

\section{S-partitions and lower bounds}

One of the main applications of red-blue pebbling is to develop lower bounds on $\texttt{OPT}_{RBP}$, i.e.\ the number of required I/O steps for specific computational DAGs from concrete applications. In both the original paper of Hong and Kung~\cite{RBpebbling1} and many follow-up works~\cite{elango2015characterizing, ranjan2012upper}, this is done via a special kind of DAG partitioning concept that allows to derive such lower bounds. In this section, we show that the same partitioning concept does not allow us to directly develop bounds that carry over to PRBP.

\subsection{S-partitions in RBP}

As one of their main tools, Hong and Kung introduce the notion of an $S$-partition of a DAG. To present this, we first define two auxiliary concepts for a given subset of nodes $V_0 \subseteq V$ in our DAG.

\begin{definition}
A subset of nodes $D \subseteq V$ is a \emph{dominator for $V_0$} if for every directed path $\pi$ that starts at a source node and ends in a node of $V_0$, it holds that $\pi$ contains a node of $D$.
\end{definition}

\begin{definition} \label{def:term}
The \emph{terminal set of $V_0$} is the set of nodes $v \in V_0$ that satisfy the following property: none of the out-neighbors of $v$ are contained in $V_0$.
\end{definition}

We note the sets in Definition~\ref{def:term} are instead called `minimum sets' in the work of~\cite{RBpebbling1}; we changed it to terminal sets since we found this original name somewhat misleading.

Given the concepts of dominators and terminal sets, an $S$-partition is defined as follows.

\begin{definition} \label{def:spart}
For a  parameter $S \! \in \! \mathbb{Z}^+$, an \emph{$S$-partition of $G$} is a disjoint partitioning $V_1$, ..., $V_k$ of the nodes of $V$ of $G$ such that
\begin{enumerate}[label=(\roman*), topsep=4pt, itemsep=1.5pt]
 \item there is no cyclic dependency among the $V_i$, i.e.\ if $u \in V_i$ and $v \in V_j$ for $i>j$, then $(u,v) \notin E$;
 \item for each $V_i$, there is a dominator set of size at most $S$ in $G$;
 \item for each $V_i$, the terminal set of $V_i$ has size at most $S$.
\end{enumerate}
\end{definition}

Hong and Kung show that if we have a valid RBP pebbling in $G$ with $r$ red pebbles, then we can use this to generate an $S$-partition of $G$ with $S=2r$ into $k$ classes such that the I/O-cost of the pebbling strategy is between $r \cdot k$ and $r \cdot (k-1)$. On a high level, this is obtained by cutting up the pebbling strategy into smaller subsequences by splitting after every $r$-th I/O operation; we then create a class $V_i$ for each subsequence, and sort each node $v \! \in \! V$ into the class of the first subsequences which places a red pebble on $v$. Intuitively, the $i$-th subsequence starts with at most $r$ red pebbles on $G$, and can load red pebbles on at most $r$ more nodes, so these $2r$ nodes are a dominator for $V_i$: all nodes computed in the subsequence indirectly come from these input values. Similarly, the $i$-th subsequence ends with at most $r$ red pebbles on $G$, and saves at most $r$ nodes to slow memory; the terminal set of $V_i$ is a subset of these $2r$ nodes. We refer the reader to~\cite{RBpebbling1} for more details.

Due to this, if $\texttt{MIN}_{part}(S)$ denotes the minimum number of classes that any $S$-partition of $G$ has, then the I/O cost in $G$ satisfies
\[ \texttt{OPT}_{RBP} \, \geq \, r \cdot (\texttt{MIN}_{part}(2r) - 1) \, . \]
As such, analyzing the $S$-partitions of concrete DAGs allows us to derive lower bounds on their I/O cost.

\subsection{S-partitions in PRBP}

Since $S$-partitions are a very efficient tool to develop lower bounds in RBP, it is natural to wonder whether the same lower bound applies to PRBP. We answer this negatively: with partial computations, the optimal cost can be much lower than the bound in the expression above.

\begin{lemma} \label{th:s-part}
There exists a construction where $\texttt{OPT}_{PRBP} = O(1)$, but we have $\texttt{MIN}_{part}(2r) = \Theta(n)$, i.e., the smallest $S$-partition according to~\cite{RBpebbling1} consists of $\Theta(n)$ classes.
\end{lemma}

\begin{proof}
Consider the DAG in Figure~\ref{fig:s-partition} with $r=3$. PRBP allows us to pebble this DAG with a trivial I/O cost of $8$: we load one of the source nodes $u_i$, then sequentially go through each node $w \! \in \! H_i$, computing $w$, marking the outgoing edge of $w$ (updating a dark red pebble on $v$), and deleting $w$. After finishing all nodes of $H_i$, we can remove the red pebble from $u_i$, and hence $3$ red pebbles suffice. Apart from loading the $u_i$ and saving $v$, this requires no I/O steps.

Now consider an $S$-partition for this DAG with $S=6$. The partition cannot consist of a single class, since with $7$ source nodes, there exists no dominator set of size $6$. Let $V'$ denote the class that contains the sink node $v$. Note that there is at least one group $H_i$ that has none of its nodes in $V'$; otherwise, there exists no dominator set of size $6$ for $V'$. The nodes in this set $H_i$ have no out-neighbors in their corresponding class (since $v \in V'$), so all these nodes in $H_i$ are contained in the terminal set of their own class. As the terminal sets can have size $6$ at most, this implies that these nodes in $H_i$ must span $|H_i| / 6$ classes at least. With $|H_i|=\Omega(n)$, the $S$-partition then consists of $\Theta(n)$ classes at least. This means that the lower bound obtained from the $S$-partition is also $\Omega(n)$.
\end{proof}
\begin{figure}
    \centering
    \begin{tikzpicture}

    \node[anchor=center] at (-10pt,57pt) {\large $u_1$};
    \node[anchor=center] at (-10pt,-63pt) {\large $u_7$};
    \node[anchor=center] at (108pt,-3pt) {\large $v$};
    \node[anchor=center] at (50pt,116pt) {group $H_1$ of};
    \node[anchor=center] at (50pt,106pt) {$\Theta(n)$ nodes};
    \node[anchor=center] at (50pt,-106pt) {group $H_7$ of};
    \node[anchor=center] at (50pt,-116pt) {$\Theta(n)$ nodes};

    \draw[very thick, gray, densely dashed] (41pt,98pt) rectangle (59pt,22pt);
    \draw[very thick, gray, densely dashed] (41pt,-98pt) rectangle (59pt,-22pt);

    \begin{scope}[thick, arrows=-stealth]
    \draw (0pt,60pt) -- (46pt,60pt);
    \draw (0pt,60pt) -- (46pt,75pt);
    \draw (0pt,60pt) -- (46.5pt,90pt);
    \draw (0pt,60pt) -- (46.5pt,30pt);
    
    \draw (0pt,-60pt) -- (46pt,-60pt);
    \draw (0pt,-60pt) -- (46pt,-45pt);
    \draw (0pt,-60pt) -- (46.5pt,-30pt);
    \draw (0pt,-60pt) -- (46.5pt,-90pt);

    \draw (50pt,60pt) -- (98.5pt,3.5pt);
    \draw (50pt,-60pt) -- (98.5pt,-3.5pt);

    \draw (75pt,0pt) -- (96pt,0pt);
    \end{scope}

    \begin{scope}[thick]
    \draw (0pt,60pt) -- (22.5pt,52.5pt);
    \draw (0pt,-60pt) -- (22.5pt,-67.5pt);

    \draw (50pt,90pt) -- (96.25pt,6.25pt);
    \draw (50pt,75pt) -- (96.25pt,6.25pt);
    \draw (75pt,22.5pt) -- (96.25pt,6.25pt);
    \draw (50pt,30pt) -- (96.25pt,6.25pt);

    \draw (50pt,-30pt) -- (96.25pt,-6.25pt);
    \draw (50pt,-45pt) -- (96.25pt,-6.25pt);
    \draw (75pt,-37.5pt) -- (96.25pt,-6.25pt);
    \draw (50pt,-90pt) -- (96.25pt,-6.25pt);
    \draw (75pt,3pt) -- (92.5pt,0pt);
    \draw (75pt,6pt) -- (92.5pt,0pt);
    \draw (75pt,-3pt) -- (92.5pt,0pt);
    \draw (75pt,-6pt) -- (92.5pt,0pt);
    \end{scope}

    \draw[black, fill=white] (100pt,0pt) circle (1.0ex);
    
    \draw[black, fill=white] (0pt,60pt) circle (1.0ex);
    
    \draw[black, fill=white] (50pt,90pt) circle (1.0ex);
    \draw[black, fill=white] (50pt,75pt) circle (1.0ex);
    \draw[black, fill=white] (50pt,60pt) circle (1.0ex);
    \node[anchor=center] at (50pt,42pt) {\large $...$};
    \draw[black, fill=white] (50pt,30pt) circle (1.0ex);

    \node[anchor=center] at (0pt,-1pt) {\Large $\mathbf{...}$};

    \draw[black, fill=white] (0pt,-60pt) circle (1.0ex);

    \draw[black, fill=white] (50pt,-30pt) circle (1.0ex);
    \draw[black, fill=white] (50pt,-45pt) circle (1.0ex);
    \draw[black, fill=white] (50pt,-60pt) circle (1.0ex);
    \node[anchor=center] at (50pt,-77pt) {\large $...$};
    \draw[black, fill=white] (50pt,-90pt) circle (1.0ex);

\end{tikzpicture}
    \caption{Construction for Lemma \ref{th:s-part}, consisting of $7$ source nodes $u_1, ..., u_7$, $7$ distinct groups $H_1, ..., H_7$ of $\Theta(n)$ nodes each, and a single sink $v$. The node $u_i$ always has edges to all the nodes in $H_i$, and all the nodes in $H_i$ have an edge towards $v$.}
    \label{fig:s-partition}
\end{figure}
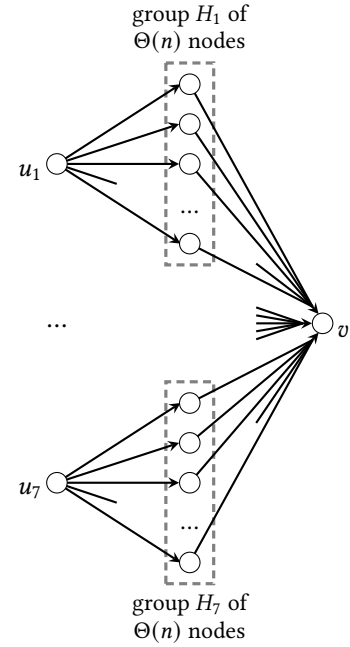

We note that the DAG in the proof has $\Delta_{in} > r$; it remains an interesting question to show a similar example with $(\Delta_{in}+1) \leq r$.

\section{Lower bound tools for PRBP} \label{sec:bounds}

The lemma above suggests that the concept of $S$-partitioning needs adjustments to provide a similarly useful tool for deriving I/O lower bounds in PRBP. In this section, we show two different ways to modify the definition of $S$-partitioning to achieve this. The first method places the focus on the edges of the DAG instead of the nodes, similarly to PRBP; this will give a general-purpose tool to derive lower bounds in PRBP, much like the original $S$-partition concept in RBP. The second one is a simpler approach that is closer to the original $S$-partition, but with one of the conditions omitted; this provides weaker lower bounds in general, but will offer a notably simpler proof for one of the concrete application graphs.

Afterwards, we revisit some relevant application graphs where I/O bounds were derived for RBP, and use our new tools to establish the same bounds with partial computations.

\subsection{$S$-edge partitions}

We first adapt our auxiliary definitions to an edge-based partitioning. For this, let $E_0 \subseteq E$ be a subset of edges in $G$.

\begin{definition}
A subset of nodes $D \subseteq V$ is an \emph{edge-dominator for $E_0$} if for every directed path $\pi$ that starts at a source node and contains an edge in $E_0$, it holds that $\pi$ contains a node of $D$.
\end{definition}

We note that if $Start(E_0)$ denotes the starting points in $E_0$, i.e.\ $Start(E_0)=\{ u \in V \, | \, \exists \, v \in V \text{ with } (u,v) \in E_0\}$, then $D \subseteq V$ is an edge-dominator for $E_0$ if and only if it is a dominator for $Start(E_0)$.

\begin{definition}
The \emph{edge-terminal set of $E_0$} is the set of nodes $v \in V$ that satisfy the following property: at least one incoming edge of $v$ is contained in $E_0$, but no outgoing edge of $v$ is in $E_0$.
\end{definition}

Note that this is a different concept from the terminal sets in Definition~\ref{def:term}: it could be that $(v_1, v_2) \! \in \! E_0$, $(v_2, v_3) \! \notin \! E_0$ and $(v_4, v_3) \! \in \! E_0$, thus both $v_2$ and its out-neighbor $v_3$ are in the edge-terminal set, which is not possible in case of terminal sets.

With these, we can already define a similar concept to $S$-partitions on the edges of our DAG.

\begin{definition} \label{def:separt}
An \emph{$S$-edge partition of $G$} is a disjoint partitioning $E_1$, ..., $E_k$ of the edges $E$ of $G$ such that
\begin{enumerate}[label=(\roman*), topsep=4pt, itemsep=1.5pt]
 \item the $E_i$ are well-ordered in the sense that for $(u,v), (v,w) \in E$ and $i<j$, we never have $(v,w) \in E_i$ and $(u,v) \in E_j$;
 \item for each $E_i$, there is an edge-dominator set of size at most $S$ in $G$;
 \item for each $E_i$, the edge-terminal set of $E_i$ has size at most $S$.
\end{enumerate}
\end{definition}

\begin{lemma} \label{lem:edge_par}
Given a DAG $G$ and an integer $r$, if there is a PRBP strategy of cost $C$ in $G$ with memory capacity $r$, then there exists a corresponding $(2r)$-edge partition of $E$ into $k$ classes such that $r \cdot k \geq C \geq r \cdot (k-1)$.
\end{lemma}

\begin{proof}
Similarly to the analysis of $S$-partitions, we split the PRBP strategy into subsequences of steps, with the $i$-th subsequence ending at the $(r\cdot i)$-th I/O operation in PRBP strategy, and the $(i+1)$-th subsequence starting immediately afterwards. If there are any more steps after the last I/O operation, we can simply add these to the last subsequence, too; these can only be delete operations that can be omitted anyway. If the strategy contains $C$ distinct I/O steps, then this results in $k=\lceil \frac{C}{r} \rceil$ subsequences. Note that every edge $(u,v)$ is marked in exactly one partial compute step in our strategy; if this compute step happens in the $i$-th subsequence (for $i \! \in \! [k]$), then we add $(u,v)$ to the class $E_i$.

The edge classes formed this way naturally fulfill property (i) above: for $(u,v), (v,w) \in E$, we can only mark $(v,w)$ in PRBP after $(u,v)$ has been marked.

For property (ii), let $V_R$ denote the set of nodes in the DAG that have a (light or dark) red pebble when the $i$-th subsequence begins, and let $V_B$ be the set of nodes on which we execute a load operation in the $i$-th subsequence. Note that $|V_R| \leq r$ and $|V_B| \leq r$; we show that $V_R \cup V_B$ is an edge-dominator (of size at most $2r$) for $E_i$. Recall that this is equivalent to $V_R \cup V_B$ being a dominator for the node set $Start(E_i)$. For an induction, consider the nodes in $Start(E_i)$ in some topological ordering. Let $u \! \in \! Start(E_i)$ be the current node in this ordering; note that there is at least one edge $(u,v) \! \in \! E_i$. Consider the pebbles on node $u$ at the beginning of the $i$-th subsequence. If $u$ has a red pebble at this point, then $u \! \in \! V_R$. If $u$ has no red pebble but it has a blue pebble, then its value needs to be loaded from slow memory to mark $(u,v)$, so $u \! \in \! V_B$. In these cases, all paths from a source to $u$ are covered by $u \! \in \! V_R \cup V_B$ itself. Finally, if $u$ has no pebble at all at the beginning of the $i$-th subsequence, then all of $u$'s in-edges are also marked within this subsequence. As such, all in-neighbors of $u$ are also in $Start(E_i)$, and hence by induction, all paths to any of these in-neighbors (and thus also all paths to $u$) are also covered by $V_R \cup V_B$.

Similarly, for property (iii), let $V'_R$ now be the set of nodes in the DAG with a red pebble at the end of the $i$-th subsequence, and $V'_B$ be the nodes on which we execute a save operation in the $i$-th subsequence. Every node $v$ in the edge-terminal set of $E_i$ has a dark red pebble at some point in the $i$-th subsequence, since we mark an incoming edge of $v$. If a dark pebble remains on $v$ until the end of the subsequence, then $v \! \in \! V'_R$. If it is removed by a save step, then $v \! \in \! V'_B$. Note that if $v$ is a sink node, we always have $v \! \in \! V'_R$ or $v \in V'_B$: if we remove the dark red pebble from it without a save operation, then we cannot have a valid pebbling. Finally, if $v$ is not a sink, then its dark red pebble could also be removed by a delete step in the same subsequence, but only after all of its out-edges are marked, too. In this case, at least one outgoing edge of $v$ is also in $E_i$, and hence $v$ is not in the edge-terminal set. Thus the edge-terminal set has size at most $2r$.
\end{proof}

We can again define $\texttt{MIN}_{edge}(S)$ as the minimal number of classes in any $S$-edge partition of $G$, which then gives us the following lower bound on I/O cost.

\begin{theorem} \label{th:dom_edge}
For the optimal cost in PRBP, we have
\[ \texttt{OPT}_{PRBP} \, \geq \, r \cdot (\texttt{MIN}_{edge}(2r) - 1) \, . \]
\end{theorem}

\subsection{$S$-dominator partitions}

For the second approach, we revisit a partitioning concept from~\cite{RBpebbling1}.
\begin{definition}
An \emph{$S$-dominator partition} of $G$ is a disjoint partitioning $V_1$, ..., $V_k$ of $V$ that fulfills properties (i) and (ii) of Definition~\ref{def:spart}, but not necessarily property (iii).
\end{definition}
Note that this is a weaker concept with less requirements than an $S$-partition; as such, if $\texttt{MIN}_{dom}(S)$ denotes the minimum number of classes that an $S$-dominator partition of $G$ must have, then $\texttt{MIN}_{dom}(S) \leq \texttt{MIN}_{part}(S)$, and as such, these kind of partitions yield looser I/O lower bounds.

Nonetheless, $S$-dominator partitions still have relevance in our case. In particular, we have seen in Lemma~\ref{th:s-part} that unlike RBP, a PRBP strategy does not necessarily generate an $S$-partition. However, we show that it still always generates an $S$-dominator partition, and thus similarly to before, we can establish a lower bound.

\begin{theorem} \label{th:dom_lower}
For the optimal cost in PRBP, we have
\[ \texttt{OPT}_{PRBP} \, \geq \, r \cdot (\texttt{MIN}_{dom}(2r) - 1) \, . \]
\end{theorem}

The claim follows directly from the lemma below.

\begin{lemma}
Given a DAG $G$ and an integer $r$, if there is a PRBP strategy of cost $C$ in $G$ with memory capacity $r$, then there exists a corresponding $(2r)$-dominator partition of $V$ into $k$ classes such that $r \cdot k \geq C \geq r \cdot (k-1)$.
\end{lemma}

\begin{proof}
Let us again divide our pebbling sequence into subsequences by splitting after each $r$-th I/O operation. For every node $v \! \in \! V$, consider the last partial compute step that marks an in-edge of $v$, and sort $v$ into the subclass corresponding to the subsequence that contains this last compute step. These last partial compute steps clearly happen to the nodes in a topological order, so there will be no cyclic dependence among the classes. As a special case, if $v$ is a source node of the DAG, then let us place $v$ into the class of the subsequence that first executes a load operation on $v$.

It only remains to show condition (ii), i.e.\ a dominator set of size $2r$ for each $V_i$. This is very similar to the proof of property (ii) in Lemma~\ref{lem:edge_par}. For a specific $V_i$, let $V_R$ denote the set of nodes with a red pebble before the $i$-th subsequence begins, and $V_B$ be the set of nodes on which we execute a load operation within the $i$-th subsequence. We again claim that $V_R \cup V_B$ is a dominator for $V_i$; this already establishes property (ii), since $|V_R \cup V_B| \leq 2r$. For an induction, consider the nodes of $V_i$ in some topological order, and let $v$ be the current node. Since $v \! \in \! V_i$, at least one of the incoming edges of $v$ was marked during the $i$-th subsequence; this means that at one point during the subsequence, $v$ must have a dark red pebble. As before, if $v$ has a red pebble at the beginning of the subsequence, then $v \! \in \! V_R$, and if $v$ only has a blue pebble (but no red), then we have to load it from slow memory, so $v \! \in \! V_B$. In both cases, $v \! \in \! V_R \cup V_B$, so all paths to $v$ are covered by $v$ itself. On the other hand, if $v$ has no pebble at the beginning, then all in-edges of $v$ are marked during this subsequence. As such, all in-neighbors of $v$ must have a dark red pebble at some point during the subsequence, so it follows by induction that all paths to the in-neighbors (and thus also to $v$) are covered.

Finally, for the special case of source nodes, our assignment ensures that they are always contained in the corresponding $V_B$.
\end{proof}

\subsection{Lower bounds for concrete computations}

We now apply the tools above to some concrete computations.

\subsubsection{Fast Fourier Transform}

Among the structured DAGs that describe concrete computations, one of the most frequently studied is the Fast Fourier Transform (FFT). Consider some integer $m$ that is a power of $2$. The $m$-point FFT graph, also known as the butterfly graph, can be obtained recursively by taking two copies of the $\frac{m}{2}$-point FFT graph, labeling its sink nodes $u_1$, ..., $u_m$, then adding a new layer $v_1$, ..., $v_m$, and adding edges from all $u_i$ to the nodes $v_j$ with $i \equiv j  \mod \frac{m}{2}$. See Figure~\ref{fig:fft} for an illustration for $m=8$.

\begin{figure}
    \centering
    \begin{tikzpicture}

    \begin{scope}[arrows=-stealth]
    \draw (0pt,0pt) -- (47pt,0pt);
    \draw (0pt,15pt) -- (47pt,15pt);
    \draw (0pt,30pt) -- (47pt,30pt);
    \draw (0pt,45pt) -- (47pt,45pt);
    \draw (0pt,60pt) -- (47pt,60pt);
    \draw (0pt,75pt) -- (47pt,75pt);
    \draw (0pt,90pt) -- (47pt,90pt);
    \draw (0pt,105pt) -- (47pt,105pt);

    \draw (0pt,0pt) -- (47.5pt,13pt);
    \draw (0pt,15pt) -- (47.5pt,2pt);
    \draw (0pt,30pt) -- (47.5pt,43pt);
    \draw (0pt,45pt) -- (47.5pt,32pt);
    \draw (0pt,60pt) -- (47.5pt,73pt);
    \draw (0pt,75pt) -- (47.5pt,62pt);
    \draw (0pt,90pt) -- (47.5pt,103pt);
    \draw (0pt,105pt) -- (47.5pt,92pt);

    \draw (50pt,0pt) -- (97pt,0pt);
    \draw (50pt,15pt) -- (97pt,15pt);
    \draw (50pt,30pt) -- (97pt,30pt);
    \draw (50pt,45pt) -- (97pt,45pt);
    \draw (50pt,60pt) -- (97pt,60pt);
    \draw (50pt,75pt) -- (97pt,75pt);
    \draw (50pt,90pt) -- (97pt,90pt);
    \draw (50pt,105pt) -- (97pt,105pt);

    \draw (50pt,0pt) -- (97.5pt,28pt);
    \draw (50pt,15pt) -- (97.5pt,43pt);
    \draw (50pt,30pt) -- (97.5pt,2pt);
    \draw (50pt,45pt) -- (97.5pt,17pt);
    \draw (50pt,60pt) -- (97.5pt,88pt);
    \draw (50pt,75pt) -- (97.5pt,103pt);
    \draw (50pt,90pt) -- (97.5pt,62pt);
    \draw (50pt,105pt) -- (97.5pt,77pt);

    \draw (100pt,0pt) -- (147pt,0pt);
    \draw (100pt,15pt) -- (147pt,15pt);
    \draw (100pt,30pt) -- (147pt,30pt);
    \draw (100pt,45pt) -- (147pt,45pt);
    \draw (100pt,60pt) -- (147pt,60pt);
    \draw (100pt,75pt) -- (147pt,75pt);
    \draw (100pt,90pt) -- (147pt,90pt);
    \draw (100pt,105pt) -- (147pt,105pt);

    \draw (100pt,0pt) -- (147.5pt,58pt);
    \draw (100pt,15pt) -- (147.5pt,73pt);
    \draw (100pt,30pt) -- (147.5pt,88pt);
    \draw (100pt,45pt) -- (147.5pt,103pt);
    \draw (100pt,60pt) -- (147.5pt,2pt);
    \draw (100pt,75pt) -- (147.5pt,17pt);
    \draw (100pt,90pt) -- (147.5pt,32pt);
    \draw (100pt,105pt) -- (147.5pt,47pt);

    \end{scope}

    \draw[black, fill=white] (0pt,0pt) circle (0.8ex);
    \draw[black, fill=white] (0pt,15pt) circle (0.8ex);
    \draw[black, fill=white] (0pt,30pt) circle (0.8ex);
    \draw[black, fill=white] (0pt,45pt) circle (0.8ex);
    \draw[black, fill=white] (0pt,60pt) circle (0.8ex);
    \draw[black, fill=white] (0pt,75pt) circle (0.8ex);
    \draw[black, fill=white] (0pt,90pt) circle (0.8ex);
    \draw[black, fill=white] (0pt,105pt) circle (0.8ex);

    \draw[black, fill=white] (50pt,0pt) circle (0.8ex);
    \draw[black, fill=white] (50pt,15pt) circle (0.8ex);
    \draw[black, fill=white] (50pt,30pt) circle (0.8ex);
    \draw[black, fill=white] (50pt,45pt) circle (0.8ex);
    \draw[black, fill=white] (50pt,60pt) circle (0.8ex);
    \draw[black, fill=white] (50pt,75pt) circle (0.8ex);
    \draw[black, fill=white] (50pt,90pt) circle (0.8ex);
    \draw[black, fill=white] (50pt,105pt) circle (0.8ex);

    \draw[black, fill=white] (100pt,0pt) circle (0.8ex);
    \draw[black, fill=white] (100pt,15pt) circle (0.8ex);
    \draw[black, fill=white] (100pt,30pt) circle (0.8ex);
    \draw[black, fill=white] (100pt,45pt) circle (0.8ex);
    \draw[black, fill=white] (100pt,60pt) circle (0.8ex);
    \draw[black, fill=white] (100pt,75pt) circle (0.8ex);
    \draw[black, fill=white] (100pt,90pt) circle (0.8ex);
    \draw[black, fill=white] (100pt,105pt) circle (0.8ex);

    \draw[black, fill=white] (150pt,0pt) circle (0.8ex);
    \draw[black, fill=white] (150pt,15pt) circle (0.8ex);
    \draw[black, fill=white] (150pt,30pt) circle (0.8ex);
    \draw[black, fill=white] (150pt,45pt) circle (0.8ex);
    \draw[black, fill=white] (150pt,60pt) circle (0.8ex);
    \draw[black, fill=white] (150pt,75pt) circle (0.8ex);
    \draw[black, fill=white] (150pt,90pt) circle (0.8ex);
    \draw[black, fill=white] (150pt,105pt) circle (0.8ex);

\end{tikzpicture}
    \caption{Illustration of the $m$-point FFT DAG for $m=8$.}
    \label{fig:fft}
\end{figure}
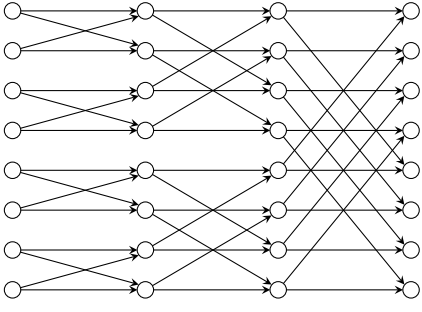

The work of~\cite{RBpebbling1} already establishes a tight lower bound of $\Omega\left( \frac{m \log m}{\log r}\right)$ on these graphs via $S$-partitions, and looser bounds are also obtained with spectral methods~\cite{jain2020spectral}. With our tools above, it is straightforward to show that these bounds also carry over to PRBP.

\begin{theorem}
For the $m$-point FFT DAG, we have
\[ \texttt{OPT}_{PRBP} \, \geq \, \Omega\left( \frac{m \log m}{\log r}\right) \, . \]
\end{theorem}

\begin{proof}
The proof of the same lower bound for RBP in~\cite{RBpebbling1} actually does not rely on an $S$-partition of the DAG, only on an $S$-dominating partition. In particular, the authors show that
\[ \texttt{MIN}_{dom}(S) \geq \Omega\left( \frac{m \log m}{S \log S} \right) \, . \]
With $S=2r$ and our Theorem \ref{th:dom_lower}, this gives the claim above.
\end{proof}

\subsubsection{Matrix Multiplication}

Matrix multiplication is a fundamental operation in computer science, and also an essential building block in various applications. Given two matrices of sizes $m_1 \times m_2$ and $m_2 \times m_3$, respectively, we consider the standard matrix multiplication algorithm which executes $m_1 \! \cdot \! m_2 \! \cdot \! m_3$ multiplications. The I/O-efficient execution of this algorithm has been extensively studied before~\cite{RBpebbling1, RBpebbling6}. Hong and Kung established a lower bound of $\Omega \left( \frac{m_1 \cdot m_2 \cdot m_3}{\sqrt{r}} \right)$ on $\texttt{OPT}_{RBP}$ using $S$-partitions. We show that the same lower bound also carries over to partial computations.

\begin{theorem} \label{th:gemm}
For standard matrix multiplication, we have
\[ \texttt{OPT}_{PRBP} \, \geq \, \Omega\left( \frac{m_1 \cdot m_2 \cdot m_3}{\sqrt{r}}\right) \, . \]
\end{theorem}

\begin{proof}
The claim can be shown by combining the original lower bound proof of~\cite{RBpebbling1} with our $S$-edge partition concept. In the corresponding DAG, we have $m_1 \cdot m_2$ sources, $m_2 \cdot m_3$ sinks, and $m_1 \cdot m_2 \cdot m_3$ internal nodes. Each internal node has exactly one outgoing edge; let us call these internal edges. We consider where the $m_1 \cdot m_2 \cdot m_3$ internal edges can be sorted in an $S$-edge partition. The edge-terminal set of each $E_i$ contains at most $S$ sink nodes, and the edge-dominator set of $E_i$ contains at most $S$ source nodes. The proof in~\cite{RBpebbling1} shows that in the matrix multiplication DAG, for any subsets of sources $V_1$ and sinks $V_2$ with $|V_1| \leq S$, $|V_2| \leq S$, there are at most $O(S^{3/2})$ internal nodes on the paths from nodes in $V_1$ to nodes in $V_2$. Thus any set of $S$ sources can be a dominator for at most $O(S^{3/2})$ internal nodes, and hence also an edge-dominator for at most $O(S^{3/2})$ internal edges in $E_i$. As for internal nodes in the edge-dominator of $E_i$, we can have at most $S$ of them, each covering at most one internal edge. As such, $E_i$ contains at most $O(S^{3/2}) + S$ internal edges. Therefore, any $S$-edge partitioning has at least $\Omega \left( \frac{m_1 \cdot m_2 \cdot m_3}{S^{3/2} + S} \right)$ classes; this is a lower bound on $\texttt{MIN}_{edge}(S)$. Selecting $S=2r$ and applying Theorem~\ref{th:dom_edge}, the claim follows.
\end{proof}

Note that the above theorem only considers the lower bound up to a constant factor, showing that its magnitude is identical in RBP and PRBP. However, recall that for the special case of $m_1=m_2$ and $m_3=1$ (matrix-vector multiplication), we have shown that we can still have $\texttt{OPT}_{PRBP} < \texttt{OPT}_{RBP}$.

\subsubsection{Flash Attention}

Due to the importance of the transformer architecture in machine learning, the I/O-efficient implementation of its crucial component, the self-attention mechanism, has been exhaustively studied. Recent work in this area also used RBP to devise an I/O lower bound that matches the prominent Flash Attention algorithm~\cite{ml1}.

While attention consists of several steps, its main bottleneck is essentially a matrix multiplication step $Q \cdot K^T$, where $Q$ and $K$ are both matrices of size $m \times d$. However, in contrast to Theorem~\ref{th:gemm}, it is a vital difference here that the entries of $Q \cdot K^T$ are not sink nodes, so we may not need to save them to slow memory. Due to this, the I/O complexity of the problem is more intricate. When we have $r \leq d^2$, a lower bound of $\Omega\left( \frac{m^2 \cdot d}{\sqrt{r}} \right)$ follows easily from matrix multiplication: intuitively, with $\frac{m^2 \cdot d}{\sqrt{r}} \geq m^2$, saving the matrix $Q \cdot K^T$ of size $m _{\!} \times _{\!} m$ to slow memory has little impact on the total cost. However, in the so-called \emph{large cache regime} when $r \geq d^2$, which is often the relevant case in practice, the Flash Attention algorithm~\cite{flashattention} can obtain a lower I/O complexity of $O\left( \frac{m^2 \cdot d^2}{r} \right)$. The work of~\cite{ml1} shows a matching lower bound; we again extend this lower bound to PRBP.

\begin{theorem}
For attention with standard matrix multiplication, we have
\[ \texttt{OPT}_{PRBP} \, \geq \, \Omega \left( \min \left( \frac{m^2 \cdot d}{\sqrt{r}} , \, \frac{m^2 \cdot d^2}{r} \right) \right) \, . \]
\end{theorem}

\begin{proof}
When $r \leq d^2$, the first term in the minimum is smaller, and the problem can be reduced to matrix multiplication as in~\cite{ml1}. The bound then carries over to PRBP due to Theorem~\ref{th:gemm}.

The second term is only smaller in the large cache regime, so assume $r \geq d^2$. In this case, similarly to Theorem~\ref{th:gemm}, we can adapt the proof of~\cite{ml1} by shifting the analysis to internal edges from internal nodes. The computational DAG of attention is described in detail in~\cite{ml1}. The relevant part of the DAG is the matrix multiplication, which is represented by $2 _{\!} \cdot _{\!} m _{\!} \cdot _{\!} d$ source nodes (entries in $Q$ and $K^T$), $m^2 _{\!} \cdot d$ internal nodes that have two of these sources as inputs, and $m \cdot m$ nodes that represent an entry in $Q \cdot K^T$ (we call these root nodes) and have incoming edges from $d$ internal nodes. Each root node also has an outgoing edge to the remaining part of the DAG (specifically, to a next node that exponentiates the value). We again refer to the edges from internal nodes to roots as internal edges, and we call a root node and its in-neighbors together an internal tree. Note that internal trees are disjoint.

On a high level, consider an $S$-edge partition of this DAG with $S=2r$, and focus on the internal edges in each class $E_i$. Both the edge-dominator and the edge-terminal set of $E_i$ has size at most $2r$, so there are at most $4r$ internal trees that contain a node form either of these sets. These trees contain at most $4r \!\cdot \! d$ internal edges altogether.

If an internal tree is not among these $4r$ trees, but it still has an internal edge that is contained in $E_i$, then we will call this an extra tree. Since an extra tree contains no nodes from the edge-terminal set of $E_i$, the outgoing edge from its root node must also be in $E_i$. The paths starting from a source and containing this edge can only be covered by the source itself, since the extra tree contains no nodes of the edge-dominator set. As such, all the $2d$ source nodes which have a path to this root must be in the edge-dominator of $E_i$.

Assume that we have $c$ extra trees, and that the $c$ entries in $Q \cdot K^T$ corresponding to their root nodes altogether span $c_i$ rows in $Q$ and $c_j$ columns in $K^T$. Then the edge-dominator of $E_i$ must contain all the $(c_{i\!}+_{\!}c_j) _{\!} \cdot _{\!} d$ source nodes in these rows and columns. Note that $c \! \leq \!  c_i \! \cdot \! c_j$ implies $\sqrt{c} \leq c_i \! + \! c_j$. Due to this, $(c_i\!+\!c_j) \! \cdot \! d \leq 2r$ implies $c \leq \frac{4 \cdot r^2}{d^2}$. The root nodes have in-degree $d$, so the $c$ extra trees contain at most $\frac{4 \cdot r^2}{d}$ internal edges in $E_i$. Together with the $4r \cdot d$ internal edges mentioned earlier, and using $r \geq d^2$, we get that $E_i$ has at most $O\left(\frac{r^2}{d} \right)$ internal edges. With altogether $m^2 \! \cdot _{\!} d$ internal edges in the DAG, we get that
\[ \texttt{MIN}_{edge}(2r) \geq \Omega \left( \frac{m^2 \cdot d^2}{r^2} \right) \, ,\]
and then the bound follows from Theorem~\ref{th:dom_edge}.
\end{proof}

\section{Finding the optimal pebbling} \label{sec:opt}

Finally, we discuss the complexity of finding the best pebbling strategy in PRBP. In particular, we show that it is NP-hard to even approximate $\texttt{OPT}_{PRBP}$ to any non-trivial factor. The same hardness result has been shown for the original RBP in recent work~\cite{mpp}; however, in order to extend this proof technique to PRBP, we need to modify the gadgets used in their construction.

\begin{theorem} \label{th:inapprox}
It is NP-hard to approximate $\texttt{OPT}_{PRBP}$ to any $n^{1-\varepsilon}$ additive term or multiplicative factor (for any $\varepsilon>0$).
\end{theorem}

\renewcommand*{\proofname}{Proof sketch}
\begin{proof}
The construction of~\cite{mpp} shows that it is already NP-hard to distinguish between an I/O cost of $2$ and an I/O cost of $n^{1-\varepsilon}$. The main ingredients in this construction are the so-called level gadgets; we need modify these so that they still ensure the same desired properties in PRBP when we have partial computations. With the level gadgets adjusted, the rest of the construction can remain unchanged, proving the same hardness result for PRBP. For more details on the original proof, we refer the reader to~\cite{mpp}.

Intuitively, each level gadget is a chain $(u_1, ..., u_{\ell})$. Between two consecutive levels $(u_1, ..., u_{\ell})$ and $(v_1, ..., v_{\ell'})$, we also have the edges $(u_i, v_i)$ for $i \leq \min(\ell, \ell')$, and if $\ell \! > \! \ell'$, then also $(u_i, v_{\ell'})$ for all $\ell' \! < \! i \! \leq \! \ell$. The DAG itself then consists of sequences of consecutive levels, called towers, with the size of each level chosen carefully. In order to adjust the construction, we will simply add \emph{auxiliary levels} between the original levels of the construction: for each level $(u_1, ..., u_{\ell})$, we add a specific number of new levels of size $\ell$ between this level and the level before. One such auxiliary level is illustrated in Figure~\ref{fig:levels}. These changes do not affect the optimum cost in RBP.

Even by adding only $1$ auxiliary level, we can already ensure that some aspects carry over to PRBP, e.g.\ establishing precedence constraints between levels. More specifically, in RBP, it is natural that if we have edges from (all nodes of) a level $L$ in one tower to (all nodes of) a level $L'$ in another tower, then we need to compute all nodes in $L$ before placing a red pebble on any node of $L'$. However, in PRBP, we could already start partially computing $L'$ before all nodes in $L$ are fully computed, which makes the analysis more difficult. In our adjusted construction, we now draw the edges from $L$ not to $L'$, but to the newly inserted auxiliary level $L'_a$ before $L'$. With partial computations, we can still place red pebbles on $L'_a$ before fully computing $L$; however, we can only start placing red pebbles on $L'$ once all incoming edges of the node below in $L'_a$ are marked. This again ensures that we can only start pebbling $L'$ after $L$ in the other tower is fully computed.

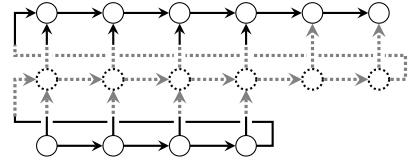
\begin{figure}
	\centering
	\begin{tikzpicture}

    \draw[thick] (75pt,0pt) -- (85pt,0pt) -- (85pt,9pt) -- (-12pt,9pt) -- (-12pt,12pt);
    \draw[thick, arrows=-stealth] (-12pt,38pt) -- (-12pt,50pt) -- (-4pt,50pt);

    \draw[very thick, densely dotted, gray] (125pt,25pt) -- (135pt,25pt) -- (135pt,34pt) -- (-12pt,34pt) -- (-12pt,38pt);

    \draw[white, fill=white] (-2pt,8pt) rectangle (2pt,10pt);
    \draw[white, fill=white] (23pt,8pt) rectangle (27pt,10pt);
    \draw[white, fill=white] (48pt,8pt) rectangle (52pt,10pt);
    \draw[white, fill=white] (73pt,8pt) rectangle (77pt,10pt);
    
    \draw[white, fill=white] (-2pt,33pt) rectangle (2pt,35pt);
    \draw[white, fill=white] (23pt,33pt) rectangle (27pt,35pt);
    \draw[white, fill=white] (48pt,33pt) rectangle (52pt,35pt);
    \draw[white, fill=white] (73pt,33pt) rectangle (77pt,35pt);
    \draw[white, fill=white] (98pt,33pt) rectangle (102pt,35pt);
    \draw[white, fill=white] (123pt,33pt) rectangle (127pt,35pt);

    \begin{scope}[thick]
    \draw (0pt,0pt) -- (0pt,11pt);
    \draw (25pt,0pt) -- (25pt,11pt);
    \draw (50pt,0pt) -- (50pt,11pt);
    \draw (75pt,0pt) -- (75pt,11pt);
    \end{scope}

    \begin{scope}[thick, arrows=-stealth]
    \draw (0pt,0pt) -- (21pt,0pt);
    \draw (25pt,0pt) -- (46pt,0pt);
    \draw (50pt,0pt) -- (71pt,0pt);

    \draw (0pt,50pt) -- (21pt,50pt);
    \draw (25pt,50pt) -- (46pt,50pt);
    \draw (50pt,50pt) -- (71pt,50pt);
    \draw (75pt,50pt) -- (96pt,50pt);
    \draw (100pt,50pt) -- (121pt,50pt);

    \draw (0pt,38pt) -- (0pt,46pt);
    \draw (25pt,38pt) -- (25pt,46pt);
    \draw (50pt,38pt) -- (50pt,46pt);
    \draw (75pt,38pt) -- (75pt,46pt);

    \end{scope}

     \begin{scope}[very thick, arrows=-stealth, densely dotted, gray]
    \draw (0pt,25pt) -- (21pt,25pt);
    \draw (25pt,25pt) -- (46pt,25pt);
    \draw (50pt,25pt) -- (71pt,25pt);
    \draw (75pt,25pt) -- (96pt,25pt);
    \draw (100pt,25pt) -- (121pt,25pt);

    \draw (0pt,12pt) -- (0pt,21pt);
    \draw (25pt,12pt) -- (25pt,21pt);
    \draw (50pt,12pt) -- (50pt,21pt);
    \draw (75pt,12pt) -- (75pt,21pt);

    \draw (100pt,25pt) -- (100pt,46pt);
    \draw (125pt,25pt) -- (125pt,46pt);

    \draw (-12pt,11pt) -- (-12pt,25pt) -- (-4pt,25pt);
     \end{scope}

    \begin{scope}[very thick, densely dotted, gray]
    \draw (0pt,25pt) -- (0pt,38pt);
    \draw (25pt,25pt) -- (25pt,38pt);
    \draw (50pt,25pt) -- (50pt,38pt);
    \draw (75pt,25pt) -- (75pt,38pt);
    \end{scope}

    \draw[black, fill=white] (0pt,0pt) circle (1.0ex);
    \draw[black, fill=white] (25pt,0pt) circle (1.0ex);
    \draw[black, fill=white] (50pt,0pt) circle (1.0ex);
    \draw[black, fill=white] (75pt,0pt) circle (1.0ex);

    \begin{scope}[thick, densely dotted, gray]
    \draw[black, fill=white] (0pt,25pt) circle (1.0ex);
    \draw[black, fill=white] (25pt,25pt) circle (1.0ex);
    \draw[black, fill=white] (50pt,25pt) circle (1.0ex);
    \draw[black, fill=white] (75pt,25pt) circle (1.0ex);
    \draw[black, fill=white] (100pt,25pt) circle (1.0ex);
    \draw[black, fill=white] (125pt,25pt) circle (1.0ex);
    \end{scope}

    \draw[black, fill=white] (0pt,50pt) circle (1.0ex);
    \draw[black, fill=white] (25pt,50pt) circle (1.0ex);
    \draw[black, fill=white] (50pt,50pt) circle (1.0ex);
    \draw[black, fill=white] (75pt,50pt) circle (1.0ex);
    \draw[black, fill=white] (100pt,50pt) circle (1.0ex);
    \draw[black, fill=white] (125pt,50pt) circle (1.0ex);

\end{tikzpicture}
	\caption{Adding an auxiliary level (dotted grey part) to the level gadgets of the construction in~\cite{mpp}. Auxiliary levels always have the same size as the original level above them.}
	\label{fig:levels}
\end{figure}

The most critical role of levels, intuitively, is to ensure that when $L=(u_1, ..., u_{\ell})$ is the last computed level in a tower, then we need to keep at least $\ell$ red pebbles in this tower (or pay an I/O cost). Let $\ell'$ be the size of the level after $L$. When we have $\ell' \! \geq \! \ell$, this property carries over easily to PRBP. The  tricky part is when $\ell' \! < \! \ell$, where PRBP would allow us to free up the red pebbles from all of $u_{\ell'+1}, ..., u_{\ell}$ by partially computing the last node of the next level, which is the only node that depends on them. To resolve this, in these cases, we add $(\ell\!-\!\ell'\!+\!2)$ auxiliary levels after $L$, and draw edges from $u_{\ell'+1}, ..., u_{\ell}$ to the last node in all of these new levels. Intuitively, this ensures that the number of nodes depending on $u_{\ell'+1}, ..., u_{\ell}$ is more than $(\ell\!-\!\ell')$, and hence we cannot free up more red pebbles with the previous approach, i.e.\ partially computing these dependent nodes and then deleting the pebbles from $u_{\ell'+1}, ..., u_{\ell}$.

Further details of the proof are discussed in Appendix~\ref{app:inapprox}.
\end{proof}
\renewcommand*{\proofname}{Proof}

\section{Conclusion and discussion}

We have seen that introducing partial computations into the red-blue pebble game can significantly change the optimal I/O cost in several computations. Even when the optimal cost remains unchanged, it necessitates new tools and theorems in order to carry over previously established I/O lower bounds to this more realistic model. We conclude with a brief discussion on some alternative variants of the model and some practical insights.

\subsection{Alternative model variants} \label{sec:models}

Except for the one-shot property, we defined and analyzed RBP according to its original definition by Hong and Kung~\cite{RBpebbling1}, which is the most well-known version of the game. However, different variants of RBP were also studied before. We briefly outline how these relate to PRBP, and discuss this in more detail in Appendix~\ref{app:models}.

If we remove the one-shot restriction, we obtain a model variant where the same node may be computed multiple times. Such re-computation steps can sometimes be used to save I/O costs, and hence they were allowed in the work of Hong and Kung and some others~\cite{RBpebbling1, RBpebbling3}. When adapting this variant to PRBP, the main modeling challenge is that re-computation could be defined in multiple different ways in a partial-computing setting. We also note that allowing re-computation can have a rather different effect on our results: for instance, Theorem~\ref{th:nphard} carries over easily to this case, Proposition~\ref{prop:reduced} can be adapted with minor changes, whereas adapting e.g.\ the tools in Section~\ref{sec:bounds} is more challenging. 

In another variant of RBP, we also allow a different compute step which  `slides' a red pebble to $v$ from one of its inputs. The sliding model aims to capture in-place computations; this is somewhat similar to PRBP, which also assumes that the inputs of $v$ are aggregated into $v$ in an in-place fashion. However, in general, partial computations allow us to save significantly more I/O cost than these sliding steps. In particular, Proposition~\ref{prop:reduced} on the large difference between $\texttt{OPT}_{RBP}$ and $\texttt{OPT}_{PRBP}$ is also easy to adapt to the case when $\texttt{OPT}_{RBP}$ is understood in the sliding model.

In another extension towards practice, we can also assign a small constant cost to the compute steps \cite{RBpebbling3, RBpebbling8}. In the one-shot model, this only adds a fixed cost to any pebbling; however, this already affects the ratio between optimal costs in e.g.\ Theorem~\ref{th:inapprox}. Note that in order to compare $\texttt{OPT}_{RBP}$ and $\texttt{OPT}_{PRBP}$ in this variant, we would need to define how the node-based compute costs in RBP are translated into edge-based partial compute costs in PRBP.

\subsection{Practical remarks and future directions}

While the usage of partial computations to reduce I/O cost has not been formalized and studied in general DAGs earlier, the general idea has of course been applied before (directly or indirectly) in several applications. For instance, we have seen how partial computations can reduce the I/O cost in a matrix-vector multiplication. In the case of dense matrix-matrix multiplication, the same approach is known as the outer-product formulation and e.g.\ used as the inner microkernel to state-of-the-art methods such as BLIS~\cite{aj3}. In a broader sense, our results also indicate that e.g.\ in generalized linear algebraic programs, the distributive property of semirings can be exploited for practical gains; the same fact has been evidenced before when exploiting re-computation to tile through successive operations (e.g.~\cite{aj2, aj1}). However, on the negative side, the inapproximability result shows that the problem is very challenging for irregular sparse computations in general.

We find that one promising direction for future work is the version of the PRBP model extended with re-computations. This model incorporates two very natural methods to reduce the I/O cost of a computation in practice, and hence it would be essential to understand its general properties, and to analyze the cost of linear algebraic or other computations in this model.

Note that another natural generalization of PRBP is to model computations that have both associative and non-associative operations, e.g.\ allowing for partial computations as long as a given order of the inputs is respected. However, such a setting can also be captured easily with a modification of the input DAG (by splitting up some of the nodes), as discussed before in~\cite{partial}.

\bibliographystyle{ACM-Reference-Format}
\bibliography{references}

\clearpage

\appendix

\section{Further proof details}

\subsection{Details for Proposition~\ref{prop:example}} \label{app:example}

For completeness, we list here the steps of the pebbling strategies that pebble the DAG with the discussed costs in our first example.

The steps of an RBP strategy with $r=4$ that only uses $3$ I/O steps is as follows: \textsc{load} $u_0$; \textsc{compute} $u_1$; \textsc{delete} $u_0$; \textsc{compute} $w_1$; \textsc{compute} $w_2$; \textsc{compute} $w_3$; \textsc{delete} $w_1$; \textsc{delete} $w_2$; \textsc{compute} $w_4$; \textsc{delete} $w_3$; \textsc{delete} $u_1$; \textsc{load} $u_0$; \textsc{compute} $u_2$; \textsc{delete} $u_0$; \textsc{compute} $v_1$; \textsc{compute} $v_2$; \textsc{delete} $w_4$; \textsc{delete} $u_2$; \textsc{compute} $v_0$; \textsc{save} $v_0$.

The steps of a PRBP strategy with $r=4$ that only uses $2$ I/O steps is as follows: \textsc{load} $u_0$; \textsc{partial compute} $(u_0, u_1)$; \textsc{partial compute} $(u_0, u_2)$; \textsc{delete} $u_0$; \textsc{partial compute} $(u_1, w_1)$; \textsc{partial compute} $(w_1, w_3)$; \textsc{delete} $w_1$; \textsc{partial compute} $(u_1, w_2)$; \textsc{partial compute} $(w_2, w_3)$; \textsc{delete} $w_2$; \textsc{partial compute} $(u_1, w_4)$; \textsc{partial compute} $(w_3, w_4)$; \textsc{delete} $u_1$; \textsc{delete} $w_3$; \textsc{partial compute} $(w_4, v_1)$; \textsc{partial compute} $(w_4, v_2)$; \textsc{partial compute} $(u_2, v_1)$; \textsc{partial compute} $(u_2, v_2)$; \textsc{delete} $w_4$; \textsc{delete} $u_2$; \textsc{partial compute} $(v_1, v_0)$; \textsc{partial compute} $(v_2, v_0)$; \textsc{save} $v_0$.

\subsection{Binary and $k$-ary trees} \label{app:k-ary}

For a more detailed discussion of the best pebbling strategies on binary trees, consider the tree on Figure~\ref{fig:gadgets} which has depth $d=3$ and $8$ leaf nodes. Let us denote the root node here $v_0$, the in-neighbors of the root $v_{1,1}$ and $v_{1,2}$ from left to right, the nodes on the third level $v_{2,1}$, ..., $v_{2,4}$ from left to right, and the leaf nodes $v_{3,1}$, ..., $v_{3,8}$ from left to right.

In RBP, computing any of the nodes $v_{2,i}$ here requires all the $r=3$ red pebbles. As such, after computing $v_{2,1}$, we need to save it to slow memory to free up the red pebbles and compute $v_{2,2}$. We can then load $v_{2,1}$ and compute $v_{1,1}$. We can delete the red pebbles from all of $v_{1,1}$, $v_{2,1}$ and $v_{2,2}$, but $v_{1,1}$ has to be saved before. We can then compute the value of $v_{1,2}$ in the same way with two more I/O operations. Finally, we have to reload $v_{1,1}$, compute $v_0$ and save $v_0$. We have used $6$ I/O operations apart from the $9$ trivial I/O steps.

In PRBP, we can compute any $v_{2,i}$ with only two red pebbles without any I/O, by loading one of the corresponding sources, partially computing the $v_{2,i}$, then deleting the source and loading the other one, and finishing the computation of $v_{2,i}$. Due to this, we can compute $v_{2,1}$ first, keep the red pebble on it and compute $v_{2,2}$ with the other two red pebbles. We can then compute $v_{1,1}$, but then we need to save it and remove its red pebble. We can then compute $v_{1,2}$ identically, load $v_{1,1}$, and compute and save $v_0$. Note that apart from the trivial $9$, we only used $2$ more I/O steps.

While $k$-ary trees are significantly less frequent in the pebbling literature, we briefly note that our observations on binary can also be generalized to this case. Once again, the most interesting behavior happens when $r=k+1$. Similarly to binary trees, the trivial I/O cost is $k^d +1$. Let us assume for convenience that $d \geq k$.

In RBP, with the exception of the bottom two levels, again for every node $v$ we have $(k-1)$ in-neighbors whose values will need to be saved and reloaded later to free up pebbles and compute the last subtree of $v$. Altogether, the number of non-trivial I/O operations in the optimal pebbling becomes
\[ 2 \cdot (k-1) \cdot \sum_{i=0}^{d-2} \, k^i = 2 \cdot (k^{d-1}-1) \, . \]

In PRBP, on the other hand, partial computations allow us to compute the bottom $r=k+1$ levels without any I/O steps. Every node $v$ above these levels will again require $2 \cdot (k-1)$ I/O steps, regardless of whether we compute $v$ partially in several iterations, or if we save $(k-1)$ of its in-neighbors and load them all later. The number of non-trivial I/O steps here is altogether
\[ 2 \cdot (k-1) \cdot \sum_{i=0}^{d-(k+1)} \, k^i = 2 \cdot (k^{d-k}-1) \, . \]

Note that the formulas above show that there is almost a factor $d^{k-1}$ difference in the non-trivial I/O cost. Altogether, adding the trivial cost of $k^d +1$, we get that $\texttt{OPT}_{RBP} = k^d + 2 k^{d-1} - 1$ and $\texttt{OPT}_{PRBP} = k^d + 2 k^{d-k}-1$.

\subsection{Proof of Lemma~\ref{lem_inset}} \label{app:clique}

We next prove that the \textsc{maxinset-vertex} problem is NP-hard. Note that any independent set of $G_0$ becomes a clique of the same size in the complement graph of $G_0$, and vice versa. As such, equivalently, we could define a \textsc{maxclique-vertex} problem that asks whether any of the cliques of maximum size in $G_0$ contain the node $v_0$. We use this clique interpretation of the problem in the proof below, since we find it more intuitive.

\begin{lemma} \label{lem_clique}
The \textsc{maxclique-vertex} problem is NP-hard.
\end{lemma}

\begin{proof}
Assume that there is a polynomial-time algorithm $\mathcal{A}$ for \textsc{maxclique-vertex}; we show how to use it to find a maximum clique in polynomial time, which is known to be NP-hard.

Given a graph $G_0$ on $n$ nodes, if all nodes have degree $(n-1)$, then the entire graph is a clique, and we are finished. Otherwise, we can invoke $\mathcal{A}$ for each node of $G_0$ separately. If there are any nodes $v$ in $G_0$ that are not included in any maximum clique, we can remove these nodes (and their incident edges). This leaves all maximum cliques intact, so we can then recursively apply our algorithm on the remaining graph. On the other hand, if all nodes are included in a maximum clique, then consider an arbitrary node $v$ with degree strictly less than $(n-1)$. This implies that there is another node $u$ in $G_0$ that is not adjacent to $v$. Note that $u$ is also part of a maximum clique, and this does not include $v$; as such, even after removing $v$, at least one of the maximum cliques in the graph remains intact. Hence we can again remove $v$ (and its incident edges), and invoke our algorithm on the remaining graph.
\end{proof}

As discussed, the NP-hardness of \textsc{maxinset-vertex} and that of \textsc{maxclique-vertex} are equivalent, so Lemma~\ref{lem_inset} follows naturally.

\subsection{Proof details for Theorem \ref{th:nphard}} \label{app:np}

Below we discuss some further details for the proof of Theorem \ref{th:nphard}. Note that the main idea of the original construction is analyzed in more detail in~\cite{RBpebbling3, mpp}.

Firstly, we note that the original construction of~\cite{RBpebbling3} does not initially use the gadget of Proposition~\ref{prop:coll} to form the groups $H_1(u_0)$ and $H_2(u_0)$, but instead just takes the same source nodes, and draws an edge from all of them to a single node. In RBP this already ensures that we need to have a red pebble on all the nodes of the gadget simultaneously. The pebble collection gadget is then only used to adapt the construction to a setting where the in-degree of the DAG is bounded by a small constant. In contrast to this, in PRBP, the use of the gadget is essential; with partial computations, adding just a singe target node would not guarantee that we need to have a red pebble on all the inputs simultaneously.

Furthermore, we note that in the construction, we also make sure that there are at least $3n_0$ further anchor nodes in each $H_1(u_0)$ and $H_2(u_0)$ that only belong to this specific gadget. These nodes are used ensure that interrupting the computation of the gadget is a suboptimal strategy. Recall that we say that a gadget $H_1(u_0)$ or $H_2(u_0)$ is visited when all the red pebbles are on the nodes of the corresponding gadget. It is then always suboptimal to visit a specific gadget twice with another gadget visited in-between, since this requires an I/O cost of at least $3n_0$ to reload the anchor nodes. If we instead go over the chain, save all the nodes in it that are included in another gadget $H_2$ (and then reload them later when needed), and then save the sink node of the chain, then this results in an I/O cost of less than $3n_0$ altogether.

The main idea of selecting the chains large enough in the gadgets has already been outlined. The actual proof is a bit more technical with two length parameters $\ell$ and $\ell_0$. Each chain starts with an initial part of length $(r-2)$ to ensure that all I/O steps in the remaining parts are non-trivial; then a long part of length $\ell_0$, then a part of length $n_0$, and then another long part of length $\ell_0$. We ensure that any of the long parts is already sufficient in itself to guarantee suboptimality if we do not place all pebbles in the gadget, i.e.\ $\frac{\ell_0}{2(r-1)}$ is already higher than the cost of any reasonable solution that pebbles each gadget in one go. Altogether, $\ell=2 \ell_0 + n_0 + (r-2)$.

We then use the $n_0$ nodes in the middle of the chain of $H_1(u_1)$ to place a distinct node into the set $H_2(u_2)$ of any node $u_2$ that is incident to $u_1$ in $G_0$. Since the degree of $u_1$ in $G_0$ is at most $n_0$, we can indeed include a separate node of the chain for each such $u_2$. Note that due to the length of the chain, we need to have all red pebbles on the gadget both before and after the computation of these $n_0$ middle nodes, otherwise the cost becomes too high. In fact, we can assume that in a reasonable strategy, when we compute such a middle node, we save it immediately to slow memory, and only load it back later when visiting the corresponding gadget $H_2(u_2)$. This only adds an I/O cost of $2$ for the given middle node. On the other hand, if we use the middle node to partially compute any other nodes in $H_2(u_2)$, then we also need to save these values and load them later when visiting $H_2(u_2)$, which again costs at least $2$.

As such, any reasonable pebbling visits each gadget at most once, computes the entire chain, and saves (and later loads) the appropriate middle nodes. In this case, it is easy to see that the gadgets indeed establish the desired dependence relations: for any edge $(u_1, u_2)$ in $G_0$, if a node in the chain of $H_1(u_1)$ is included in $H_2(u_2)$, and a node in the chain of $H_1(u_2)$ is included in $H_2(u_1)$, then we must visit $H_1(u_1)$ before $H_2(u_2)$ and we must visit $H_1(u_2)$ before $H_2(u_1)$ so at most of the the pairs $H_1(u_1)$, $H_2(u_1)$ and $H_1(u_2)$, $H_2(u_2)$ can be visited consecutively. For a simpler analysis, we can assume that we also add a similar dependence from $H_1(u_1)$ to $H_2(u_1)$, to ensure that $H_1(u_1)$ is visited before $H_2(u_1)$.

Note that each group $H_1(u_0)$ and $H_2(u_0)$ needs to have at least a given size in order to ensure that we can pick distinct nodes for the distinct roles in each group. In particular, we need $3n_0$ anchor nodes, as discussed before. We need $b$ nodes that will be merged between $H_1(u_0)$ and $H_2(u_0)$. In $H_2(u_0)$, we need a specific node from the chain of $H_2(u_1)$ for each node $u_1$ in $G_0$ that is incident to $u_0$, which can be up to $n_0$. We also need $3$ further nodes for the nodes in $Z_1$ or $Z_2$. Altogether, this leads to a choice of $r-2=|H_1(u_0)|=|H_2(u_0)|=b+4n_0+3$.

As for the rest of our parameters, we can select $b$ as a large constant; it has to be larger than $|Z_1|=|Z_2|=3$ in order to ensure that the I/O gain from the consecutive visitation of a gadget-pair is still much larger than the worst-case cost of computing $w$, and hence regardless of $w$'s role, the optimal pebbling will still always follow a maximum independent set. This defines $r=b+4n_0+5$ as discussed above. Finally, if we pebble each gadget simultaneously, but cannot visit any pair consecutively, we get a non-trivial I/O cost of at most $n_0 \cdot b$ for reloading the $b$ merged nodes in the gadget-pairs, $2 \cdot |E_0|$ for having to reload the middle nodes in the chains (where $|E_0|$ is the number of edges in $G_0$), plus $6$ for reloading $Z_1$ and $Z_2$. We need to ensure that $\frac{\ell_0}{2(r-2)}$ is larger than the sum of this. Furthermore, we add another $r-1$ here, to ensure that even the non-trivial I/O cost incurred by the gadget is larger than this sum. Thus altogether, we want to guarantee
\[ \frac{\ell_0}{2(r-2)} - (r-1) > n_0 \cdot b + 2 \cdot |E_0| + 6 \, , \]
which can be satisfied with a choice of $\ell_0 = \Theta(r \cdot (n_0 + |E_0| + r))=\Theta(n_0^2 + n_0 \cdot |E_0|)$. Recall that then $\ell=2 \ell_0 + n_0 + (r-2)$, which is in the same magnitude.

Note that with this choice of parameters, the size of the whole construction is still polynomial in the size of the original graph $G_0$.

\subsection{Proof of Theorem \ref{th:inapprox}} \label{app:inapprox}

We then discuss further details for the proof of Theorem \ref{th:inapprox}. Recall that in order to adapt this proof, we only add a given number of auxiliary levels before each original level in the construction of~\cite{mpp}. The remaining part of the proof (the towers formed from the gadgets, and the edges connecting them) is unchanged.

Recall that a level gadget is a chain $(u_1, ..., u_{\ell})$ with $(u_i, u_{i+1}) \! \in \! E$ for $i \! \in \![ \ell\!-\!1 ]$. If the level $(u_1, ..., u_{\ell})$ is followed by the level $(v_1, ..., v_{\ell'})$, then $(u_i, v_i)\! \in \! E$ for $i \! \in \! [\min(\ell, \ell')]$. If $\ell \! > \! \ell'$, then also $(u_i, v_{\ell'}) \! \in \! E$ for $\ell' \! < \! i \! \leq \! \ell$. For an illustration of the gadgets, see Figure 3 in the full version of~\cite{mpp}. Our auxiliary levels are simply extra levels between the original levels of the construction, which always have the size of the next original level that follows. Adding these new levels does not affect the behavior in RBP at all: we can still transition between the original levels in the same way as before, just with more intermediate steps. As long as the total number of auxiliary levels is polynomial in the DAG size $n$, the whole reduction remains polynomial.

In contrast to RBP, if have have a level $(u_1, ..., u_{\ell})$ in PRBP, it is possible that the red pebble from e.g.\ the node $u_1$ is already removed before we place a red pebble on $u_{\ell}$, and as such, it is not easy to show that an optimal solution will always pebble a level in a consecutive fashion. Thus in the analysis for PRBP, we will instead say that our strategy is currently on a given original level $(u_1, ..., u_{\ell})$ of a tower if we already placed a red pebble on all of $u_1, ..., u_{\ell}$ at least once, but the same does not hold for the following level.

With this modified definition, and adding at least $1$ auxiliary level before each original level, we can already carry over the precedence relations established between different towers to PRBP. Whenever a level $L$ in a tower had edges towards some level $L'$ of another tower, we now draw those edges to the corresponding nodes of the auxiliary level $L'_a$ before $L'$. While we can now do partial computations to already place red pebbles on $L'_a$ before having reached $L$, we can only place red pebbles on the nodes of $L'$ once the incoming edges of a node below in $L'_a$ are already all marked, and hence we can only reach level $L'$ after reaching level $L$ in the other tower.

Recall that in general, if $L=(u_1, ..., u_{\ell})$ and $L'=(v_1, ..., v_{\ell'})$ are two consecutive levels in a tower, we also needs to ensure that if we are on level $L$, then we need to have at least $\ell$ pebbles in this tower also in PRBP. This holds easily when $\ell' \! \geq \! \ell$: for all indices $i \! \in \! [\ell]$, before reaching $L'$, a valid pebbling always has a pebble on the $i$-th node of either $L$, $L'$, an auxiliary level in-between, or the level after $L'$ (let us add an auxiliary level on the top of each tower to ensure that this argument also easily extends to the last levels). However, for $\ell' \! < \! \ell$, we need our extra modification here: we add $(\ell\!-\!\ell'\!+_{\!}2)$ auxiliary levels before $L'$ (with incoming edges from other towers still going to the lowermost one), and draw an edge from all of $u_{\ell'+1}, ..., u_{\ell}$ to the last node of each of these auxiliary levels. This ensures that partially computing the last nodes of each auxiliary level between $L$ and $L'$ (and hence possibly removing the red pebbles from $u_{\ell'+1}, ..., u_{\ell}$ before satisfying the external dependencies of $L'$) offers no advantage, since these last nodes together lock down more than $(\ell\!-\!\ell')$ pebbles. For the sake of precision, we can change our definition for these levels, and say that we are already on level $L'$ when having placed a red pebble on all nodes of the second auxiliary level; then once again, any pebbling has at least $\ell$ pebbles in the given tower while on level $L$.

As another technical detail, note that the last level $L$ of the main tower has incoming edges from the last level $L'$ in any other tower. To ensure that we cannot save pebbles by partially computing the values of $L$, we add a separate auxiliary level of $L$ for the incoming edges of each such $L'$. With $L$ being larger than any of the $L'$, this means that before reaching the penultimate level in the main tower, partially computing the auxiliary levels of $L$ (to free up pebbles from $L'$) only locks down more pebbles, and thus offers no advantage.

\section{Other model variants} \label{app:models}

We discuss a few more details on introducing partial computations to some other variants of RBP that appear in the literature.

\subsection{Re-computations}

An RBP variant with re-computation can be obtained by simply removing the one-shot restriction, i.e.\ allowing a node $v$ to be computed any number of times. This can be a beneficial strategy in many DAGs; if e.g.\ the inputs of a node $v$ are anyway kept in fast memory for some reason, and $v$ is needed later in our computation, then instead of saving $v$ to slow memory and loading it back later (I/O cost of $2$), we can simply delete the red pebble from $v$ and later compute $v$ again from its inputs at no I/O cost.

This variant also has the undesired effect that the optimal pebbling sequence might grow super-polinomially long because we repeatedly delete and re-compute nodes at no cost, making the problem PSPACE-complete~\cite{RBpebbling2}. Due to this, many works instead focus on an RBP variant that is is NP, such as the one-shot model.

Note that introducing re-computation into PRBP is a delicate matter, since the nodes are computed in multiple steps. One simple solution is to only allow re-computing the value of the node from scratch, i.e.\ we again need to aggregate all the inputs of $v$. This model extension is relatively easy to define; we briefly discuss its properties below. Alternatively, a more advance model could allow to only repeat some of the partial compute steps that $v$ consists of, and this can indeed be beneficial in some cases. In particular, it can happen that after aggregating most of the inputs into a node $v$, we are forced to save the almost-final value of $v$ into slow memory and load it later. Assume that at this point, $v$ has a final unmarked in-edge $(u,v)$, and two outputs $w_1$, $w_2$. Furthermore, assume that between the computation of $w_1$ and $w_2$, we are forced to remove the red pebble from $v$ again, but the red pebble on $u$ is still kept. In the original PRBP, we can mark $(u,v)$, then mark $(v,w_1)$, but then we have to save the final value of $v$ to slow memory and load it back later, at a cost of $2$. However, if we are allowed instead to keep the unfinished value of $v$ in slow memory, and re-compute only the edge $(u,v)$, then we could instead mark $(u,v)$ and $(v,w_1)$, delete the red pebble from $v$, and later when computing $w_2$, simply load the almost-finalized value of $v$ from slow memory again, and aggregate $u$ into $v$ again, at a cost of only $1$. Executing this sequence of steps is also possible in practice, and hence a PRBP extension could also decide to follow this approach, and allow to re-compute only specific parts of $v$. We leave it to future work to analyze the properties of such a model.

We briefly discuss a PRBP extension where we are only allowed to re-compute a node from scratch. For this, we add a new operation:
\begin{enumerate}[label=\arabic*), leftmargin=1.5em]
\setcounter{enumi}{4}
 \item \textsc{clear}: given a node $v$ that is neither a source nor a sink, remove all pebbles from $v$ and unmark all the in-edges of $v$.
\end{enumerate}
We could also require here that all in-edges of $v$ are already marked, since otherwise any partial compute step on $v$ before the clearing was useless, and hence the pebbling could be simplified. Similarly, the restriction that $v$ is neither a source or a sink may be removed.

The DAG in Figure~\ref{fig:decreased_opt} provides a nice example where re-computation allows us to save I/O operations in RBP. In particular, with $r=4$ here, we can load $u_0$, compute $u_1$, $w_1$ and $w_2$, then delete $u_1$, compute $w_3$, then re-compute $u_1$, compute $u_2$, and finish the pebbling as in the original model. This results in $\texttt{OPT}_{RBP}=2$.

To adapt this example to RBP with re-computation, we can add another layer of two nodes $z_1$, $z_2$ after $u_0$, i.e.\ $u_0$ now has edges to only $z_1$, $z_2$, while $z_1$, $z_2$ both have edges to $u_1$, $u_2$. Intuitively, this ensures that we need to keep two red pebbles in fast memory in order to re-compute $u_1$, which is not possible while using $3$ red pebbles to compute $w_3$. The same approach can be used to adapt Proposition~\ref{prop:reduced}: our modified gadgets now span from $z_1$, $z_2$ until $v_1$, $v_2$. Note that PRBP was incurring only the trivial I/O cost anyway, so it remains unaffected; this settles Propositions~\ref{prop:example} and~\ref{prop:reduced}.

In contrast to this, some of our other claims carry over to this extended model without changes. In Proposition~\ref{prop:mat_vec}, the internal nodes all have outdegree $1$, so re-computation does not decrease the cost here in either RBP or PRBP. With source nodes not being re-computable, the optimal pebbling strategies in the DAGs of Figure~\ref{fig:gadgets} also remain unchanged. With the construction of Theorem~\ref{th:nphard} consisting of pebble collection gadgets, re-computation also does not affect the optimal cost in this case. We note here that the original works of ~\cite{RBpebbling3, mpp} consider yet another RBP variant where source nodes do not have to be loaded, but can be computed freely anytime. Due to this, their work adds further gadgets in front of the sources to discourage the re-computation of these values. Since we work with the original RBP variant of~\cite{RBpebbling1} where source nodes are not computed, this modification is not required in our case.

Regarding $S$-partitions, Lemma~\ref{th:s-part} of course holds for PRBP with re-computation, but the rest of our observations are more challenging to adapt. In particular, when converting a pebbling strategy to a partition, we simply assigned the nodes/edges into the subsequence where they were computed. The same assignment rule becomes more technical when values can be computed multiple times. While the base ideas of adapting these concepts to re-computation have already been outlined for standard $S$-partitions in~\cite{RBpebbling1}, we leave it to future work to present and analyze the adaptation of the tools and proofs in Section~\ref{sec:bounds} to re-computations.

Similarly, proof construction for Theorem~\ref{th:inapprox} in \cite{mpp} was only designed for one-shot RBP, so establishing inapproximabilty results for PRBP with re-computation is again a non-trivial task.

\subsection{Sliding pebbles}

In RBP with sliding pebbles, we also allow another kind of compute step: if all inputs of a node $v$ have a red pebble, then we can move the red pebble from one of the inputs to $v$. This RBP model was rarely analyzed on its own, but it is often mentioned because many results for the so-called black pebble game are developed in a variant that allows sliding pebbles, and these can often be transferred to RBP if we allow sliding.

One motivation for this RBP variant is to capture in-place computations: intuitively, we can think of sliding as the node $v$ being computed in the same memory location that was storing one of its inputs before. Due to this, combining this version with partial computations is less interesting: in PRBP, the partial compute step essentially already assumes that the next value is combined with the aggregated value in an in-place way. The only exception to this is when the very first input is aggregated into $v$, which would only require $1$ red pebble in an in-place interpretation, but needs $2$ red pebbles simultaneously in PRBP; however, this does not play a major role for most DAGs.

An interesting question, however, is how RBP compares to PRBP if we allow sliding pebbles in RBP, or on a more abstract level, how much of the cost difference between RBP and PRBP is due to the in-place nature of computations, and how much is due to the ability to compute nodes partially? We briefly revisit our claims in Section~\ref{sec:io_relate} to show that in-place computation plays a minor role only, and instead partial computations are the relevant factor.

Note that the DAG in Figure~\ref{fig:decreased_opt} is actually an example where sliding pebbles allow us to also have $\texttt{OPT}_{RBP}=2$, by sliding a red pebble from $w_1$ to $w_3$ when computing $w_3$. However, the example can be easily adjusted to this model version by adding a third node $w_0$ which has an incoming edge from $u_1$ and an outgoing edge to $w_3$. In this new graph, even the sliding model needs $3$ red pebbles simultaneously for pebbling $w_3$, so our original arguments hold again; on the other hand, PRBP can still pebble the DAG wihtout any extra I/O. The same modification also covers Proposition~\ref{prop:reduced}.

The proof of Proposition~\ref{prop:mat_vec} holds again if we restrict ourselves to $r \geq 2m-1$. In the zipper gadget and pebble collection gadget, we obtain the same example behaviors as before if we choose $r=d+1$ instead. Adapting the pebble collection gadget also covers the case of Theorem~\ref{th:nphard} if $r$ is decreased by $1$.

The only significant difference is the case of binary trees, where sliding pebbles actually also allow us to compute each $3$-node subtree with only $2$ red pebbles, and hence the optimal number of non-trivial I/O steps in RBP also becomes $2^{d-1}-2$, similarly to PRBP. As such, Proposition~\ref{prop:bintree} does not carry over to the case when RBP allows sliding pebbles. Note that this is only the case for binary trees; in e.g.\ a ternary tree, a $4$-node subtree can still be computed with $2$ red pebbles (without extra I/O) in PRBP, but already requires $3$ red pebbles in RBP with pebble sliding. As such, for $k$-ary trees with $k \geq 3$, we again have $\texttt{OPT}_{PRBP} < \texttt{OPT}_{RBP}$.

\subsection{Computation costs}

In RBP with computation costs, we assume that compute steps also incur a cost of some $\varepsilon>0$~\cite{RBpebbling3, RBpebbling8, mpp}. This parameter $\varepsilon$ is typically considered a constant, and much smaller than the cost $1$ of the I/O operations. This model guarantees that $\texttt{OPT}_{RBP} \geq \varepsilon \cdot n$, which means that our claims considering the ratio of two costs, such as the multiplicative part of Theorem~\ref{th:inapprox} or the linear factor in Proposition~\ref{prop:reduced}, do not carry over to this variant.

However, there is also a fundamental modeling question when trying to compare $\texttt{OPT}_{RBP}$ and $\texttt{OPT}_{PRBP}$ in this model, namely, how compute costs are adapted to PRBP. If we simply assign a cost of $\varepsilon$ to every partial compute step, then (in the one-shot model) this will add up to a cost of $\varepsilon \cdot |E|$ in PRBP, in contrast to $\varepsilon \cdot n$ in RBP, and hence the two costs are not directly comparable anymore. Another option is to assign a cost of $\varepsilon / \text{deg}_{in}(v)$ to each partial compute step on an edge $(u,v)$, where $\text{deg}_{in}(v)$ is the in-degree of $v$. This leaves the total costs identical, but it means that different partial compute steps have different weights, which might not be desired.

Note that in general, this variant leads to a different model interpretation where we are not only trying to capture I/O costs, but the `total cost' of executing the computation. This has limited relevance in one-shot RBP where the total computation cost is always $\varepsilon \cdot n$, but it is very natural in other settings, e.g.\ with re-computation if we want to understand the trade-off between the extra compute steps and the I/O cost saved, or in parallel versions of RBP~\cite{mpp} where it offers a natural way to motivate a balance between distributing workload and minimizing communication costs.

\subsection{No deletion}

Another RBP variant assumes that red pebbles can never be deleted, but they instead get replaced by the blue pebble when saved to slow memory~\cite{RBpebbling2, RBpebbling3}. This extension has less practical relevance, and is instead rather a tool to ensure that the optimal pebbling sequence has polynomial length (as opposed to RBP with re-computation, where this is not guaranteed). PRBP could be adapted to this setting by changing the deletion rule such that a dark red pebble may never be removed from a node by any other means than saving it. Similarly to RBP, this ensures that $\texttt{OPT}_{PRBP} \geq n-r$ in any DAG, since all nodes are saved at least once to slow memory, except for possibly the $r$ nodes that have a red pebble in the final state.

\end{document}